\documentclass[11pt]{article}%

\usepackage{amsfonts}
\usepackage{amsmath}
\usepackage{fullpage}
\usepackage{amssymb}
\usepackage{graphicx}
\usepackage{hyperref}
\usepackage{amsfonts}
\usepackage{amsmath}
\usepackage{fullpage}
\usepackage{amssymb}
\usepackage{graphicx}
\usepackage{hyperref}
\usepackage{color}%

\providecommand{\U}[1]{\protect\rule{.1in}{.1in}}

\newtheorem{theorem}{Theorem}

\newtheorem{lemma}[theorem]{Lemma}

\newtheorem{remark}[theorem]{Remark}

\newenvironment{proof}[1][Proof]{\noindent\textbf{#1.} }{\ \rule{0.5em}{0.5em}}
\numberwithin{equation}{section}

\def\hil{{\mathcal H}}
\def\kil{{\mathcal K}}

\def\B{{\mathcal B}}

\def\D{{\mathcal D}}
\def\E{\mathcal{E}}

\def\N{{\mathcal N}}
\def\P{{\mathcal P}}

\def\R{{\mathcal R}}
\def\S{{\mathcal S}}

\def\iff{\Longleftrightarrow}
\def\imp{\Longrightarrow}
\def\ep{\varepsilon}
\def\bN{\mathbb{N}}
\def\bC{\mathbb{C}}
\def\bR{\mathbb{R}}

\def\bz{\left(}
\def\jz{\right)}

\def\ad{\operatorname{ad}}
\def\prod{\mathrm{pr}}

\def\qed{{\tiny\ensuremath\blacksquare}}

\def\wt{\widetilde}
\def\ps{p_{\operatorname{succ}}}
\def\deff{\equiv}

\def\x{^{x}}

\newcommand{\s}{\mbox{ }}
\newcommand{\ds}{\mbox{ }\mbox{ }}
\newcommand{\norm}[1]{\left\| #1\right\|}
\newcommand{\cbnorm}[1]{\left\| #1\right\|_{\mathrm{CB},1\to\alpha}}

\newcommand{\diad}[2]{|#1\rangle\langle #2|}
\newcommand{\pr}[1]{\diad{#1}{#1}}
\newcommand{\ket}[1]{|#1\rangle}

\DeclareMathOperator{\id}{id}
\DeclareMathOperator{\Tr}{Tr}
\DeclareMathOperator{\supp}{supp}

\begin{document}

\title{Strong converse exponents for a quantum channel discrimination problem and
quantum-feedback-assisted communication}
\author{Tom Cooney\thanks{Hearne Institute for Theoretical Physics, Department of
Physics and Astronomy, Louisiana State University, Baton Rouge, Louisiana
70803, USA}
\and Mil\'an Mosonyi\thanks{F\'isica Te\`orica: Informaci\'o i Fen\`omens
Qu\`antics, Universitat Aut\`onoma de Barcelona, E-08193 Bellaterra
(Barcelona), Spain } \thanks{Mathematical Institute, Budapest University of
Technology and Economics, Egry J\'ozsef u 1., Budapest, 1111 Hungary}
\and Mark M. Wilde\footnotemark[1] \thanks{Center for Computation and Technology,
Louisiana State University, Baton Rouge, Louisiana 70803, USA } }
\maketitle

\begin{abstract}
This paper studies the difficulty of discriminating between an arbitrary
quantum channel and a \textquotedblleft replacer\textquotedblright\ channel
that discards its input and replaces it with a fixed state. The results
obtained here generalize those known in the theory of quantum hypothesis
testing for binary state discrimination. We show that, in this particular
setting, the most general adaptive discrimination strategies provide no
asymptotic advantage over non-adaptive tensor-power strategies. This
conclusion follows by proving a quantum Stein's lemma for this channel
discrimination setting, showing that a constant bound on the Type~I error
leads to the Type~II error decreasing to zero exponentially quickly at a rate
determined by the maximum relative entropy registered between the channels.
The strong converse part of the lemma states that any attempt to make the
Type~II error decay to zero at a rate faster than the channel relative
entropy implies that the Type~I error necessarily converges to one. We then
refine this latter result by identifying the optimal strong converse exponent
for this task. As a consequence of these results, we can establish a strong
converse theorem for the quantum-feedback-assisted capacity of a channel,
sharpening a result due to Bowen. Furthermore, our channel discrimination
result demonstrates the asymptotic optimality of a non-adaptive tensor-power
strategy in the setting of quantum illumination, as was used in prior work on
the topic. The sandwiched R\'{e}nyi relative entropy is a key tool in our analysis.
Finally, by combining our results with recent results of Hayashi and Tomamichel,
we find a novel operational interpretation of the mutual information of a quantum channel
$\mathcal{N}$
as the optimal type II error exponent when discriminating between a large number
of independent instances of $\mathcal{N}$ and an arbitrary ``worst-case'' replacer channel chosen
from the set of all replacer channels.

\end{abstract}

\section{Introduction}

Quantum channel discrimination is a natural extension of a basic problem in
quantum hypothesis testing, that of distinguishing between the possible states
of a quantum system. In the case of binary state discrimination, it is given
\textit{a priori} that a quantum system is in one of two states $\rho$
or~$\sigma$, and the goal is to identify in which state it is by performing a
quantum measurement. We say that $\rho$ is the null hypothesis and $\sigma$ is
the alternative hypothesis. A natural extension of this problem occurs in the
independent and identically distributed (i.i.d.) setting. Here, the
discriminator is provided with $n$ quantum systems in the state $\rho^{\otimes
n}$ or $\sigma^{\otimes n}$, and the task is to apply a binary measurement
$\{Q_{n},I^{\otimes n}-Q_{n}\}$ on these $n$ systems, with $0\leq Q_{n}\leq
I^{\otimes n}$, to determine which state he possesses. One is then concerned
with two kinds of error probabilities:%
\begin{equation}
\alpha_{n}(Q_{n})\equiv\Tr\left\{  (I^{\otimes n}-Q_{n})\rho^{\otimes
n}\right\}  ,\label{alphan}%
\end{equation}
the probability of incorrectly rejecting the null hypothesis, the Type~I
error, and%
\begin{equation}
\beta_{n}(Q_{n})\equiv\Tr\left\{  Q_{n}\sigma^{\otimes n}\right\}
,\label{betan}%
\end{equation}
the probability of incorrectly rejecting the alternative hypothesis, the
Type~II error. Of course, it is generally impossible to find a quantum
measurement such that both of these errors are equal to zero simultaneously,
so one instead studies the asymptotic behaviour of $\alpha_{n}$ and $\beta
_{n}$ as $n\rightarrow\infty$, expecting there to be a trade-off between
minimising $\alpha_{n}$ and minimising $\beta_{n}$.

In asymmetric hypothesis testing, one fixes a constraint on the Type~I error,
say, and then seeks to minimise the Type~II error. When a constant threshold
$\varepsilon$ is imposed on the Type~I error, the optimal Type~II error is
given by
\begin{align}
\beta_{\varepsilon}(\rho\|\sigma)\equiv\min\{\beta(Q):\,0\le Q\le
I,\,\alpha(Q)\le\varepsilon\}.
\end{align}
The central result in the asymptotic setting is the quantum Stein's lemma, due
to Hiai and Petz \cite{HP} and Ogawa and Nagaoka \cite{ON}. The direct part of
the lemma states that for any constant bound on the Type~I error, there exists
a sequence of measurements $\{Q_{n},I^{\otimes n}-Q_{n}\}$ that meets this
constraint and is such that the Type~II error decreases to zero exponentially
fast with a decay exponent given by the quantum relative entropy $D(\rho
\Vert\sigma)$,
defined as \cite{Umegaki,HP}
\begin{equation}
D(\rho\Vert\sigma)\equiv\left\{
\begin{array}
[c]{cc}%
\text{Tr}\left\{  \rho\left[  \log\rho-\log\sigma\right]  \right\}  & \text{if
supp}\left(  \rho\right)  \subseteq\text{supp}\left(  \sigma\right) \\
+\infty & \text{otherwise}%
\end{array}
\right.  . \label{eq:vn-rel-ent}%
\end{equation}
In the above and throughout the paper, we take the logarithm to be base two.
Furthermore, the strong converse part of the lemma states that any attempt to
make the Type~II error decay to zero with a decay exponent larger than the
relative entropy will result in the Type~I error converging to one in the
large~$n$ limit \cite{ON}. The direct and the strong converse parts can be
succinctly written as
\begin{align}
\label{Stein's lemma}\lim_{n\to\infty}-\frac{1}{n}\log\beta_{\varepsilon}%
(\rho^{\otimes n}\|\sigma^{\otimes n})=D(\rho\|\sigma),\ds\ds\ds \forall\ep\in(0,1).
\end{align}
That is, for any threshold value $\varepsilon\in(0,1)$, the optimal Type~II
error decays exponentially fast in the number of copies, and the decay rate is
equal to the relative entropy.

It is easy to see that the negative logarithm of the optimal Type~II error,
\begin{align}
\label{eq:htre}D_{H}^{\varepsilon}(\rho\|\sigma)\equiv-\log\beta_{\varepsilon
}(\rho\|\sigma),
\end{align}
is non-negative and monotonic non-increasing under completely positive
trace-preserving maps. Thus, it can be considered as a \textquotedblleft
generalized divergence\textquotedblright\, or \textquotedblleft generalized
relative entropy\textquotedblright\, and it was named \textquotedblleft
hypothesis testing relative entropy\textquotedblright\ in \cite{WR12}. With
this notation, Stein's lemma \eqref{Stein's lemma} can be reformulated as
\begin{equation}
\lim_{n\rightarrow\infty}\frac{1}{n}D_{H}^{\varepsilon}(\rho^{\otimes n}%
\Vert\sigma^{\otimes n})=D(\rho\Vert\sigma),\ds\ds\ds \forall\ep\in(0,1). \label{eq:steins-htre}%
\end{equation}

As a refinement of the quantum Stein's lemma, one can study the optimal Type~I
error given that the Type~II error decays with a given exponential speed. One
is then interested in the asymptotics of the optimal Type~I error
\begin{align}
\alpha_{n,r}\equiv\alpha_{2^{-nr}}(\rho^{\otimes n}\|\sigma^{\otimes n}%
)\equiv\min\left\{ \alpha_{n}(Q_{n}):\,0\le Q_{n}\le I,\,\beta_{n}(Q_{n}%
)\le2^{-nr}\right\} ,
\end{align}
with $r>0$ a constant. In the \textquotedblleft direct
domain,\textquotedblright\ when $r<D(\rho\Vert\sigma)$, $\alpha_{n,r}$ also
decays with an exponential speed, as was shown in \cite{OH}. The exact decay
rate is determined by the quantum Hoeffding bound theorem
\cite{PhysRevA.76.062301,N06,ANSV} as
\begin{equation}
\lim_{n\rightarrow\infty}-\frac{1}{n}\log\alpha_{n,r}=H_{r}(\rho
\|\sigma)\equiv\sup_{0<\alpha<1}\frac{\alpha-1}{\alpha}(r-D_{\alpha}(\rho
\Vert\sigma)), \label{eq:direct-part-err-exp}%
\end{equation}
where $D_{\alpha}$ is a quantum R\'enyi relative entropy, to be defined later,
and $H_{r}(\rho\|\sigma)$ is the Hoeffding divergence of $\rho$ and
$\sigma$. On the other hand, in the \textquotedblleft strong converse
domain,\textquotedblright\ when $r>D(\rho\Vert\sigma)$, $\alpha_{n,r}$ goes to
$1$ exponentially fast \cite{ON,N}. The rate of this convergence has been
determined in \cite[pages 80-81]{Hayashibook} in terms of the limit of
post-measurement R\'enyi relative entropies. A ``single-letter'' expression has been
obtained recently in \cite{MO} as
\begin{equation}
\lim_{n\rightarrow\infty}-\frac{1}{n}\log(1-\alpha_{n,r})=H_{r}^{*}%
(\rho\|\sigma)\equiv\sup_{\alpha>1}\frac{\alpha-1}{\alpha}(r-\widetilde
{D}_{\alpha}(\rho\Vert\sigma)), \label{eq:sc-states}%
\end{equation}
where $\tilde D_{\alpha}$ is an alternative version of the quantum R\'enyi
relative entropy \cite{MDSFT13,WWY13}, and $H_{r}^{*}(\rho\|\sigma)$ is the
Hoeffding anti-divergence. Note that it is unique to the quantum case that
one requires a R\'{e}nyi relative entropy for the strong converse domain which
is different from that used in the direct domain (however, these R\'{e}nyi
relative entropies coincide when $\rho$ and $\sigma$ commute, i.e., the
classical case).

The results in \eqref{eq:direct-part-err-exp} and \eqref{eq:sc-states} give a
complete understanding of the trade-off between the two error probabilities in
the asymptotics. Note that the quantum Stein's lemma can also be recovered
from \eqref{eq:direct-part-err-exp} and \eqref{eq:sc-states} in the limit
$r\to D(\rho\|\sigma)$. We remark that there are other ways of refining our
understanding of the quantum Stein's lemma, as established recently in
\cite{li12,TH12}. \medskip

The objectives of channel discrimination are very similar to those of state
discrimination; what makes the problem different is the complexity of the
available discrimination strategies. In the general setup we have a quantum
channel with input system $A$ and output system $B$, and we know that the
channel is described by either ${\mathcal{N}}_{1}$ or ${\mathcal{N}}_{2}$, where 
${\mathcal{N}}_{1}$ and ${\mathcal{N}}_{2}$ are completely positive trace-preserving (CPTP) maps.
We assume that we can use the channel several times, consecutive uses are
independent, and the properties of the channel do not change with time. Thus,
$n$ uses of the channel are described by either ${\mathcal{N}}_{1}^{\otimes
n}$ or ${\mathcal{N}}_{2}^{\otimes n}$. A non-adaptive discrimination strategy
for $n$ uses of the channel consists of feeding an input state $\psi
_{R_{n}A^{n}}$ into the $n$-fold tensor-product channel, and then performing a
binary measurement $\{Q_{n},I-Q_{n}\}$
on the output, which is either ${\mathcal{N}}_{1}^{\otimes n}(\psi_{R_{n}%
A^{n}}) \equiv
(\id_{R_n} \otimes {\mathcal{N}}_{1}^{\otimes n})(\psi_{R_{n}%
A^{n}})$ or ${\mathcal{N}}_{2}^{\otimes n}(\psi_{R_{n}%
A^{n}}) \equiv(\id_{R_n} \otimes {\mathcal{N}}_{2}^{\otimes n})(\psi_{R_{n}A^{n}})$. Here, $R_{n}$
is an ancilla system on which the channel acts trivially as the identity map
$\id_{R_n}$. When an adaptive
strategy is used, the output of the first $k$ uses of the channel can be used
to prepare the input for the $(k+1)$-th use; see Figure
\ref{FigureChannelDiscrim1} for a pictorial explanation and Section
\ref{sec:channel Stein} for a precise definition.

For any discrimination strategy $S_{n}$, let $\rho_{n}(S_{n})$ and $\sigma
_{n}(S_{n})$ denote the output of the $n$-fold product channel depending on
whether the channel is equal to ${\mathcal{N}}_{1}$ or ${\mathcal{N}}_{2}$. In
analogy with \eqref{alphan}-\eqref{betan}, one can define the Type~I and the
Type~II errors as
\begin{align}\label{error probs}
\alpha_{n}(S_{n})\equiv\Tr \{(I-Q_{n})\rho_{n}(S_{n}%
)\},\mbox{ }\mbox{ }\text{(Type~I)}%
\mbox{ }\mbox{ }\mbox{ }\mbox{ }\mbox{ }\mbox{ }\mbox{ }\mbox{ }\mbox{ }\mbox{ }
\beta_{n}(S_{n})\equiv\Tr \{Q_{n}\sigma_{n}(S_{n}%
)\},\mbox{ }\mbox{ }\text{(Type~II)},
\end{align}
where $\{Q_{n},I-Q_{n}\}$ is the measurement part of the strategy. It is then
natural to consider the optimal error probabilities
\begin{align}
\beta_{\varepsilon}\x({\mathcal{N}}_{1}^{\otimes n}\|{\mathcal{N}}%
_{2}^{\otimes n}) & \equiv\inf\{\beta_{n}(S_{n}):\,\alpha_{n}(S_{n}%
)\le\varepsilon
\},\mbox{ }\mbox{ }\mbox{ }\mbox{ }\mbox{ }\mbox{ }\mbox{ }\mbox{ }\mbox{ }\mbox{ }\mbox{ }\mbox{ }\mbox{ }\mbox{ }\text{and}%
\\
\alpha_{n,r}\x\equiv\alpha_{2^{-nr}}\x({\mathcal{N}}_{1}^{\otimes
n}\|{\mathcal{N}}_{2}^{\otimes n}) & \equiv\inf\left\{ \alpha_{n}%
(S_{n}):\,\beta_{n}(S_{n})\le2^{-nr}\right\} ,\label{channel alpha}%
\end{align}
where $x$ denotes the set of allowed discrimination strategies and the
optimisations are over all strategies in the class $x$. Here, we will consider
$x=\ad$ for adaptive and $x=\prod$ for product strategies. The
latter are all non-adaptive strategies with an input state $\psi_{R_{n}A^{n}%
}=\psi_{RA}^{\otimes n}$, where $\psi_{RA}$ is an arbitrary state on $A$ and
some ancilla $R$. Obviously, if only product strategies are allowed ($x = \prod$), then the
optimal rates of these error probabilities are given by the corresponding
channel divergences as
\begin{align}
\lim_{n\to+\infty}-\frac{1}{n}\log\beta_{\varepsilon}\x({\mathcal{N}%
}_{1}^{\otimes n}\|{\mathcal{N}}_{2}^{\otimes n}) & = D({\mathcal{N}}%
_{1}\|{\mathcal{N}}_{2})\equiv\sup_{\psi_{RA}}D({\mathcal{N}}_{1}(\psi
_{RA})\|{\mathcal{N}}_{2}(\psi_{RA})),\label{channel Stein}\\
\lim_{n\to+\infty}-\frac{1}{n}\log\alpha_{n,r}\x & =H_{r}({\mathcal{N}%
}_{1}\|{\mathcal{N}}_{2})\equiv\sup_{\psi_{RA}}H_{r}({\mathcal{N}}_{1}%
(\psi_{RA})\| {\mathcal{N}}_{2}(\psi_{RA})),\label{channel Hoeffding}\\
\lim_{n\to+\infty}-\frac{1}{n}\log(1-\alpha_{n,r}\x) & =H_{r}%
^{*}({\mathcal{N}}_{1}\|{\mathcal{N}}_{2})\equiv\inf_{\psi_{RA}}H_{r}%
^{*}({\mathcal{N}}_{1}(\psi_{RA})\| {\mathcal{N}}_{2}(\psi_{RA}%
)),\label{channel HK}%
\end{align}
according to the previously explained results on state discrimination. Note
that in \eqref{channel HK} an infimum is taken; the reason is that, in the
strong converse domain, the goal is to minimise the exponent of the success
probability. The Hoeffding (anti-)divergences can also be
expressed as
\begin{align}
H_{r}({\mathcal{N}}_{1}\|{\mathcal{N}}_{2}) & =\sup_{0<\alpha<1}\frac
{\alpha-1}{\alpha}(r-D_{\alpha}({\mathcal{N}}_{1}\Vert{\mathcal{N}}_{2})), \\
H_{r}^{*}({\mathcal{N}}_{1}\|{\mathcal{N}}_{2}) & =\sup_{1<\alpha}\frac
{\alpha-1}{\alpha}(r-\widetilde D_{\alpha}({\mathcal{N}}_{1}\Vert{\mathcal{N}%
}_{2})) ,
\end{align}
where $D_{\alpha}({\mathcal{N}}_{1}\Vert{\mathcal{N}}_{2})$ and $\widetilde
D_{\alpha}({\mathcal{N}}_{1}\Vert{\mathcal{N}}_{2})$ are the channel R\'enyi
relative entropies:
\begin{align}
D_{\alpha}({\mathcal{N}}_{1}\Vert{\mathcal{N}}_{2}) &  \equiv\sup_{\psi_{RA}%
}D_{\alpha}({\mathcal{N}}_{1}(\psi_{RA})\|{\mathcal{N}}_{2}(\psi_{RA})),\label{channel Renyi1}\\
\widetilde D_{\alpha}({\mathcal{N}}_{1}\Vert{\mathcal{N}}_{2})  &  \equiv
\sup_{\psi_{RA}}\widetilde D_{\alpha}({\mathcal{N}}_{1}(\psi_{RA}%
)\|{\mathcal{N}}_{2}(\psi_{RA})).\label{channel Renyi2}
\end{align}
The optimizations in \eqref{channel Stein}--\eqref{channel HK}
and \eqref{channel Renyi1}--\eqref{channel Renyi2} are taken over all possible bipartite states
$\psi_{RA}$ with an arbitrary ancilla system $R$.

For adaptive strategies, the relations
\eqref{channel Stein}-\eqref{channel HK} are not expected to hold for
arbitrary channels. For instance, the results of \cite{HHLW} provide some
evidence in this direction. (See~\cite{PhysRevLett.103.210501} as well for
related results and more general conclusions.) There are various classes of
channels, however, for which \eqref{channel Stein}-\eqref{channel HK} hold;
for these channels, adaptive strategies do not offer any benefit over product
strategies. For instance, Hayashi showed
\eqref{channel Stein}-\eqref{channel HK} with $x=\ad$ for any pair of
classical channels \cite{H09}.

Another extreme case is when both ${\mathcal{N}}_{1}$ and ${\mathcal{N}}_{2}$
are replacer channels, i.e., there exist states $\rho,\sigma$ such that
${\mathcal{N}}_{1}(\cdot)={\mathcal{R}}_{\rho}\equiv\Tr\{\cdot\}\rho$ and
${\mathcal{N}}_{2}(\cdot)={\mathcal{R}}_{\sigma}\equiv\Tr\{\cdot\}\sigma$.
Obviously, in this case all the channel divergences are equal to the
corresponding divergences of the two states; e.g., $D_{\alpha}({\mathcal{R}%
}_{\rho}\|{\mathcal{R}}_{\sigma})=D_{\alpha}(\rho\|\sigma)$, etc. It is also
heuristically clear that adaptive strategies do not offer any benefit over
product strategies, and the channel discrimination problem reduces to the
state discrimination problem between $\rho$ and $\sigma$, described before.
Two of our main results, Theorems \ref{stein} and \ref{scrtheorem} yield as a
special case a mathematically precise argument for these heuristics in the
case of
\eqref{channel Stein} and \eqref{channel HK}.

A natural intermediate step towards determining the error exponents of the
general quantum channel discrimination problem is to allow one of the channels
to be arbitrary, while keeping the other channel a replacer channel. This
setup interpolates between the fully understood case of state discrimination
and the still open problem of general quantum channel discrimination. Here we
consider the setup in which the first channel is arbitrary and the second
channel is a replacer channel. We prove \eqref{channel Stein} (Stein's lemma)
in Section \ref{sec:Stein proof}, and show in Section \ref{sec:sc proof} that
the strong converse exponent is given as in \eqref{channel HK} for adaptive
strategies ($x=\ad$). As for now, we leave the optimality part of
\eqref{channel Hoeffding} open for $x=\ad$.

As a consequence of these results, in Section \ref{sec:EA sc proof} we can
establish a strong converse theorem for the quantum-feedback-assisted capacity
of a channel, which is the capacity of a quantum channel for transmitting classical information
with the assistance of a noiseless quantum feedback from receiver to sender.
Our result here strengthens that of Bowen's \cite{B04}. 
We also make a connection between our results and quantum
illumination \cite{L} in Section \ref{sec:quantum-illum}. Finally, in Section~\ref{sec:extension},
we discuss how to combine the recent results in \cite{HT} with ours to obtain a
quantum Stein's lemma in a setting more general than that considered in either paper.
This gives a novel operational interpretation of the mutual information of a quantum channel,
different from that already found in entanglement- and quantum-feedback-assisted communication
\cite{ieee2002bennett,Hol01a,B04}.
We also discuss an open question regarding the characterization of the strong converse
exponent in this more general setting.

\section{Summary of results}

\subsection{Quantum Stein's lemma in adaptive channel discrimination}

\label{sec:channel Stein}

Our first result is a generalization of the quantum Stein's lemma in
(\ref{eq:steins-htre}) to the setting of adaptive quantum channel
discrimination. In particular, we study the difficulty of discriminating
between an arbitrary quantum channel ${\mathcal{N}}$ and a \textquotedblleft
replacer\textquotedblright\ channel ${\mathcal{R}}$ that discards its input
and replaces it with a fixed state $\sigma$. An important physical realization
of this problem is in quantum illumination \cite{L,PhysRevLett.101.253601}
(discussed more in Section~\ref{sec:quantum-illum}). We show that a
tensor-power strategy is optimal in this case, so that there is no need to
consider the most general adaptive strategy (at least in the asymptotic
regime). This can be seen as a quantum Stein's lemma for this task; if one
optimises the Type~II error under the constraint that the Type~I error is less
than some fixed constant $\varepsilon\in(0,1)$, then the optimal Type~II error
probability cannot decrease to zero exponentially faster than a rate
determined by the relative entropy. Otherwise, the Type~I error necessarily
converges to one. It is straightforward to employ the direct part of the
established quantum Stein's lemma from \cite{HP} in order to establish the
direct part for our setting.\begin{figure}[ptb]
\begin{center}
\includegraphics[width=6in]{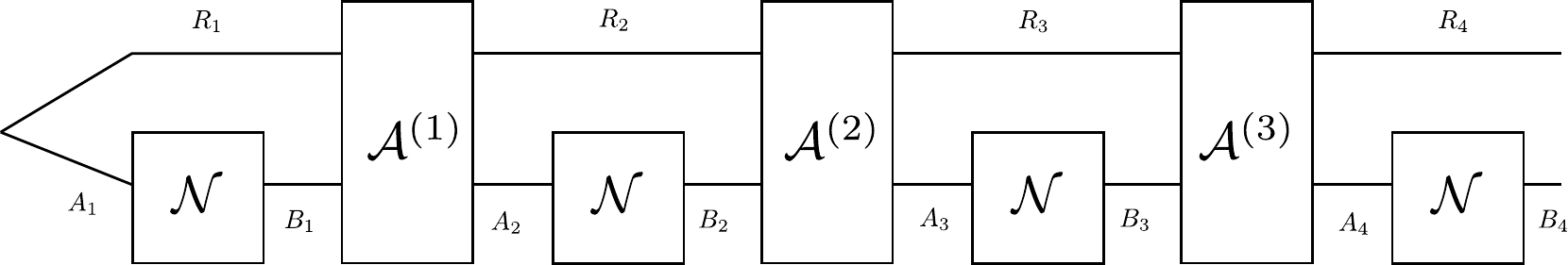}
\end{center}
\caption{A four-round adaptive discrimination strategy applied to the channel
${\mathcal{N}}$.}%
\label{FigureChannelDiscrim1}%
\end{figure}

\begin{figure}[ptb]
\begin{center}
\includegraphics[width=6in]{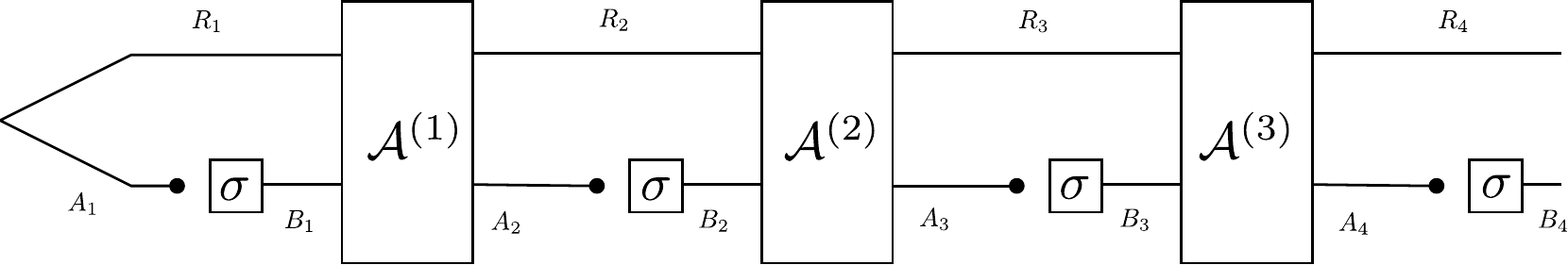}
\end{center}
\caption{A four-round adaptive discrimination strategy applied to the replacer
channel ${\mathcal{R}}$.}%
\label{FigureChannelDiscrim2}%
\end{figure}

In more detail, the most general adaptive discrimination strategy is depicted
in Figures~\ref{FigureChannelDiscrim1} and~\ref{FigureChannelDiscrim2}. It
consists of a choice of input state $\rho_{R_{1}A_{1}}$, a sequence
$\{\mathcal{A}_{R_{i}B_{i}\rightarrow R_{i+1}A_{i+1}}^{(i)}\}_{i\in
\{1,\ldots,n-1\}}$ of adaptive quantum channels, and finally a quantum
measurement $\left\{  Q_{R_{n}B_{n}},I_{R_{n}B_{n}}-Q_{R_{n}B_{n}}\right\}  $
to decide which channel was applied. Let $\tau_{R_{n}B_{n}}$ denote the output
state at the end of the adaptive discrimination strategy (before the final
measurement $\left\{  Q_{R_{n}B_{n}},I_{R_{n}B_{n}}-Q_{R_{n}B_{n}}\right\}  $
is performed) when the channel being applied is ${\mathcal{R}}$, and let
$\rho_{R_{n}B_{n}}$ denote the output state at the end of the adaptive
discrimination strategy when the channel being applied is ${\mathcal{N}}$. Let
$D_{H,\ad}^{\varepsilon}\left(  \mathcal{N}^{\otimes n}\Vert\mathcal{R}^{\otimes
n}\right)  $ denote the \textquotedblleft adaptive hypothesis testing relative
entropy,\textquotedblright\ which generalizes \eqref{eq:htre} by allowing for
an optimization over all possible adaptive strategies used to discriminate
between $\mathcal{N}^{\otimes n}$ and
$\mathcal{R}^{\otimes n}$. We
define it formally as follows:%
\begin{equation}
D_{H,\ad}^{\varepsilon}\left(  \mathcal{N}^{\otimes n}\middle\Vert
\mathcal{R}^{\otimes n}\right)  \equiv-\log\beta_{\varepsilon}^{\ad
}({\mathcal{N}}^{\otimes n}\|{\mathcal{R}}^{\otimes n})= -\log\inf\Tr\left\{
Q_{R_{n}B_{n}}\tau_{R_{n}B_{n}}\right\}  , \label{eq:adaptive-htre}%
\end{equation}
where the infimum is over all measurement operators $Q_{R_{n}B_{n}}$ subject
to $0\leq Q_{R_{n}B_{n}}\leq I_{R_{n}B_{n}}$ and%
\begin{equation}
\Tr\left\{  Q_{R_{n}B_{n}}\rho_{R_{n}B_{n}}\right\}  \geq1-\varepsilon,
\end{equation}
all preparation states $\rho_{R_{1}A_{1}}$ subject to $\rho_{R_{1}A_{1}}\geq0$
and Tr$\left\{  \rho_{R_{1}A_{1}}\right\}  =1$, and all adaptive quantum
channels
\begin{equation}
\left\{  \mathcal{A}_{R_{i}B_{i}\rightarrow R_{i+1}A_{i+1}}^{\left(  i\right)
}\right\}  _{i\in\left\{  1,\ldots,n-1\right\}  }.
\end{equation}

We can now state our first main result:

\begin{theorem}
\label{stein} Let $\varepsilon\in\left(  0,1\right)  $ be a fixed constant.
Let ${\mathcal{N}}:\mathcal{B}({\mathcal{H}}_{A})\rightarrow\mathcal{B}%
({\mathcal{H}}_{B})$ be an arbitrary quantum channel and let ${\mathcal{R}%
}:\mathcal{B}({\mathcal{H}}_{A})\rightarrow\mathcal{B}({\mathcal{H}}_{B})$ be
the replacer quantum channel ${\mathcal{R}}(X_{A})=\Tr\{X_{A}\}\sigma_{B}$,
for some fixed density operator $\sigma_{B}$. Then the channel version of
Stein's lemma, \eqref{channel Stein} holds, i.e.,
\begin{align}
\lim_{n\rightarrow\infty}-\frac{1}{n}\log\beta_{\varepsilon}^{\ad
}\left(  \mathcal{N}^{\otimes n}\middle\Vert\mathcal{R}^{\otimes n}\right)  &=
\lim_{n\rightarrow\infty}\frac{1}{n}D_{H,\ad}^{\varepsilon}\left(
\mathcal{N}^{\otimes n}\middle\Vert\mathcal{R}^{\otimes n}\right)  \\
&=
\sup_{\psi_{RA}}D\left(  \mathcal{N}_{A\rightarrow B}(  \psi_{RA})
\Vert\psi_{R}\otimes\sigma_{B}\right)  =
D(  \mathcal{N}
\Vert\mathcal{R}) \label{eq:adaptive-steins}%
\end{align}
for any $\varepsilon\in(0,1)$. It suffices to take system $R$
isomorphic to system $A$ in the above optimization.
\end{theorem}

This theorem clearly generalizes the quantum Stein's lemma in
(\ref{eq:steins-htre}). It implies that a tensor-power discrimination strategy
is optimal---allowing for an adaptive strategy yields no asymptotic
improvement. That is, one should simply prepare $n$ copies of the bipartite
state $\psi_{RA}$ optimizing (\ref{eq:adaptive-steins}), send each $A$ system
through each channel use\ (creating the state $\left[  \mathcal{N}%
_{A\rightarrow B}(  \psi_{RA})  \right]  ^{\otimes n}$ or $\left[
\mathcal{R}_{A\rightarrow B}(  \psi_{RA})  \right]  ^{\otimes n}$),
and finally perform a collective measurement on all systems $R^{n}B^{n}$ to
decide which channel was applied.

\subsection{The strong converse exponent for adaptive channel discrimination}

Next, we refine our analysis by identifying the strong converse exponent for
the task of discriminating between an arbitrary quantum channel ${\mathcal{N}%
}$ and a replacer channel ${\mathcal{R}}$. 
It is easy to see (by considering $\ep\to 0$) that Theorem \ref{stein} implies that
for any rate $r<D(\N\|\R)$, there exists a sequence of non-adaptive strategies, along which 
the type I error goes to zero, and the type II error vanishes exponentially fast, with a rate at least $r$. 
This is usually referred to as the direct part of Stein's lemma.
Moreover, it also implies that the strong converse property holds, i.e., for any 
sequence of adaptive srategies, if the 
type II error vanishes exponentially with a rate $r>D(\N\|\R)$, then the type I error
goes to $1$ (this can be seen by taking $\ep\to 1$).
Our aim is to determine the speed of convergence of the type I error to $1$ in the strong converse domain, for any decay rate 
$r>D(\N\|\R)$ of the type II errors. As it turns out, this convergence is also exponential, and hence our aim is to 
determine the exact values of the strong converse exponents:
%
\begin{align}
\mathrm{\underline{sc}}(r) & \equiv\inf\left\{  \liminf_{n\rightarrow+\infty
}-\frac{1}{n}\log\Tr\left\{  Q_{R_{n}B_{n}}\rho_{R_{n}B_{n}}\right\}
:\,\liminf_{n\rightarrow+\infty}-\frac{1}{n}\log\Tr\left\{  Q_{R_{n}B_{n}}%
\tau_{R_{n}B_{n}}\right\}  >r\right\} ,\label{sci def}\\
\mathrm{\overline{sc}}(r) & \equiv\inf\left\{  \limsup_{n\rightarrow+\infty
}-\frac{1}{n}\log\Tr\left\{  Q_{R_{n}B_{n}}\rho_{R_{n}B_{n}}\right\}
:\,\liminf_{n\rightarrow+\infty}-\frac{1}{n}\log\Tr\left\{  Q_{R_{n}B_{n}}%
\tau_{R_{n}B_{n}}\right\}  >r\right\}  ,\label{scs def}%
\end{align}
where the infimum is over all sequences of adaptive measurement strategies,
specified by measurement operators $Q_{R_{n}B_{n}}$ subject
to $0\leq Q_{R_{n}B_{n}}\leq I_{R_{n}B_{n}}$, preparation states
$\rho_{R_{1}A_{1}}$, and adaptive quantum channels
\begin{equation}
{\mathcal{A}}_{[n]}\equiv\left\{  \mathcal{A}_{R_{i}B_{i}\rightarrow
R_{i+1}A_{i+1}}^{\left(  i\right)  }\right\}  _{i\in\left\{  1,\ldots
,n-1\right\}  }.
\end{equation}


We establish the following theorem:
\begin{theorem}\label{scrtheorem} 
Let ${\mathcal{N}}:\mathcal{B}({\mathcal{H}}_{A})\rightarrow\mathcal{B}({\mathcal{H}}_{B})$ 
be an arbitrary quantum channel, and let ${\mathcal{R}}:\mathcal{B}({\mathcal{H}}_{A})\rightarrow
\mathcal{B}({\mathcal{H}}_{B})$ be the replacer quantum channel ${\mathcal{R}%
}(X)=\Tr\{X\}\sigma_{B}$, for some fixed density operator $\sigma_{B}$. For any
$r>D(\N\|\R)$,
\begin{align}
\mathrm{\underline{sc}}(r)=\mathrm{\overline{sc}}(r)& =\lim_{n\to+\infty}%
-\frac{1}{n}\log(1-\alpha_{n,r}^{\ad})\label{scrtheroem0}  \\
& =\sup_{\alpha>1}\inf
_{\psi_{RA}}\frac{\alpha-1}{\alpha}\left[  r-\widetilde{D}_{\alpha}(
{\mathcal{N}}_{A\rightarrow B}(\psi_{RA})\Vert\psi_{R}\otimes\sigma
_{B})  \right] \label{scrtheroem1}\\
&  =\inf_{\psi_{RA}}\sup_{\alpha>1}\frac{\alpha-1}{\alpha}\left[
r-\widetilde{D}_{\alpha}(  {\mathcal{N}}_{A\rightarrow B}(\psi
_{RA})\Vert\psi_{R}\otimes\sigma_{B})  \right]
\label{scrtheroem2}\\
&=\sup_{\alpha>1}\frac{\alpha-1}{\alpha}\left[  r-\widetilde{D}_{\alpha}(
\mathcal{N}\Vert\mathcal{R})  \right] ,\label{scrtheroem3}
\end{align}
where $\alpha_{n,r}^{\ad}$ is defined in \eqref{channel alpha},
the infima are taken over all possible bipartite states
$\psi_{RA}$ with an arbitrary ancilla system $R$;
in particular, \eqref{channel HK} holds. 
Moreover, the same identities hold when the infima are restricted
to pure states $\psi_{RA}$ with $R$ being a fixed copy of $A$.
\end{theorem}

\begin{remark}
When $D(\N\|\R)=+\infty$, the above statement is empty. On the other hand,
when $D(\N\|\R)$ is finite, then Theorem 2 also holds for $0<r\le D(\N\|\R)$ in a trivial way.
 Indeed, by Theorem \ref{stein}, 
if $r\le D(\N\|\R)$ then
the operational quantities in 
\eqref{scrtheroem0} are equal to $0$, and so is \eqref{scrtheroem3}, since $\widetilde D_{\alpha}(\N\|\R)\ge D(\N\|\R)$ for every $\alpha>1$, according to Lemma \ref{MI limits}.
\end{remark}

\subsection{Connection to quantum illumination}

\label{sec:quantum-illum}

Our results have implications for the theory of quantum illumination, which we
discuss briefly here. Building on prior work in
\cite{PhysRevA.72.014305,PhysRevA.71.062340}, Lloyd \textit{et al}.~show how
the use of entangled photons can provide a significant improvement over
unentangled light when detecting the presence of an object
\cite{L,PhysRevLett.101.253601}. The goal in quantum illumination is to
determine whether a distant object is present or not by employing quantum
light along with a quantum detection strategy. It is sensible and traditional
\cite{L,PhysRevLett.101.253601}\ to take the object not being present as the
null hypothesis and the object being present as the alternative hypothesis.

In the usual scenario, the transmitter and receiver are in the same location.
The protocol begins with the transmitter sending a signal mode that is
entangled with an idler mode still in the possession of the transmitter. Let
$\left\vert \psi\right\rangle _{SI}$ denote the state of the signal and idler
mode. If the object is not present (the null hypothesis), then the signal mode
is lost and is replaced by a thermal state~$\theta_{S}$, so that the joint
state becomes $\theta_{S}\otimes\psi_{I}$. Clearly, this is an instance of the
replacer channel. If the object is present (the alternative hypothesis), then
the signal beam is reflected off the object and returns to the transmitter.
The resulting state is described by $\left(  \mathcal{N}_{S}\otimes
\text{id}_{I}\right)  \left(  \psi_{SI}\right)  $, where $\mathcal{N}_{S}$
describes the noise characteristics of the reflection channel. This protocol
is performed $n$ times with the receiver storing either the state $\left[
\theta_{S}\otimes\psi_{I}\right]  ^{\otimes n}$ or $\left[  \left(
\mathcal{N}_{S}\otimes\text{id}_{I}\right)  \left(  \psi_{SI}\right)  \right]
^{\otimes n}$. The receiver finally performs a collective measurement on all
of the systems in order to decide whether the object is present. Thus, we have
a quantum channel discrimination problem in which one seeks to distinguish
between a replacer channel and a noisy channel. However, our results do not
apply to this setting if one takes the null and alternative hypotheses in the
natural way suggested above.

An alternative scenario is that in which the transmitter and receiver are in
different locations. It is technologically more challenging to take advantage
of quantum illumination in this setting, due to the fact that the transmitter
and receiver need to share and store entanglement over a potentially large
distance. Nevertheless, this is the setting to which our results apply. Given
that the null hypothesis in this setting corresponds to the object not being
present, the channel applied to the transmitted mode will be $\mathcal{N}$,
which characterizes the optical loss in the transmission. Since the
alternative hypothesis in this setting corresponds to the object being present
and such an object will reflect the light incident on it, the signal beam does
not make it to the receiving end and the receiver instead detects thermal
noise, so that the channel applied to the transmitted mode is the replacer
channel $\mathcal{R}$. Thus, the Type~I and Type~II errors for this setting
correspond to our setting described in the previous sections.

Implicit in prior analyses on quantum illumination is the assumption that a
tensor-power, non-adaptive strategy is optimal. Our results support this
assumption (at least in the particular setting of asymmetric hypothesis
testing described above) by showing that no asymptotic advantage is provided
by instead using an adaptive strategy for quantum channel
discrimination.\footnote{Strictly speaking, the results in our paper apply to
finite-dimensional systems, whereas the quantum illumination protocols apply
to infinite-dimensional, albeit finite-energy, systems. Given that our
analysis never has any dimension dependence, this suggests that it should be possible
to extend our results to infinite-dimensional systems with energy constraints.} It remains an
open question to determine if a tensor-power, non-adaptive strategy is optimal
in the symmetric hypothesis testing setting considered in
\cite{L,PhysRevLett.101.253601}.

\subsection{Strong converse theorem for quantum-feedback-assisted
communication}
\label{sec:EA sc}

There is a well-known connection between hypothesis testing and channel
coding, first recognized by Blahut \cite{B74}, and this connection also holds
for quantum channels. The direct part of the channel coding theorem (i.e., the
Holevo-Schumacher-Westmoreland theorem) \cite{H, SW} can be obtained from the
direct part of Stein's lemma, as shown in \cite{HN, ON07}.

One consequence of Theorem~\ref{stein} is a strong converse theorem for the
quantum-feedback-assisted classical capacity of a quantum channel. In prior
work, Bowen proved that a noiseless quantum feedback channel does not increase
the entanglement-assisted capacity
\cite{PhysRevLett.83.3081,ieee2002bennett,Hol01a} of a noisy channel, by
proving a weak converse for its quantum-feedback-assisted capacity \cite{B04}.
That is, Bowen proved that the quantum-feedback-assisted capacity of a channel
$\mathcal{N}$ is equal to its entanglement-assisted capacity, denoted by%
\begin{equation}\label{IN}
I(\mathcal{N})\equiv\sup_{\psi_{RA}}\inf_{\sigma_{B}}D\left(  \mathcal{N}%
_{A\rightarrow B}(  \psi_{RA})  \Vert\psi_{R}\otimes\sigma
_{B}\right)  .
\end{equation}
However, Bowen's result did not exclude the possibility of a trade-off between
the communication rate and the error probability; our strong converse theorem
shows that no such trade-off is possible in the asymptotic limit of many
channel uses. A strong converse theorem in this context states that for any
coding scheme, which seeks to transmit at a rate strictly higher than the
capacity of the channel, the probability of successful decoding
decays to zero exponentially fast in the number of channel uses. So our result
sharpens Bowen's \cite{B04}, strengthens the main result of \cite{GW}, and
generalizes \cite[Theorem~7]{PV} to the quantum case. The approach taken is
inspired by that used by Nagaoka \cite{N}, who derived the strong converse
theorem for any memoryless quantum channel from the monotonicity of the
R\'enyi relative entropies. Polyanskiy and Verd\'{u} \cite{PV} later
generalised this approach to show how a bound on the success probability could
be derived from any relative-entropy-like quantity that satisfies certain
natural properties. This approach has already been used to prove several
strong converse theorems for quantum channels \cite{KW, WWY13, GW, TWW14};
here we shall use the sandwiched R\'{e}nyi relative entropy \cite{MDSFT13,
WWY13}.

With the proof of Theorem~\ref{stein} in hand, it requires only a little extra effort to
prove a strong converse for the quantum-feedback-assisted capacity of a
quantum channel (the capacity when unlimited use of a noiseless quantum
feedback channel from receiver to sender is allowed).

\begin{theorem}
\label{thm:sc-feedback}Let $p_{\operatorname{succ}}$ denote the success
probability of any rate $R$ quantum-feedback-assisted communication code for a
channel ${\mathcal{N}}$ that uses it $n\geq1$ times. The following bound holds%
\begin{equation}
p_{\operatorname{succ}}\leq2^{-n\sup_{\alpha>1}\left(  \frac{\alpha-1}{\alpha
}\right)  \left(  R-\widetilde{I}_{\alpha}\left(  \mathcal{N}\right)  \right)
}, \label{eq:QFA-strong-conv-exponent}%
\end{equation}
where $\widetilde{I}_{\alpha}\left(  \mathcal{N}\right)  $ is the sandwiched
R\'{e}nyi mutual information of the channel $\mathcal{N}$, defined in
\eqref{EA4}. As a consequence of this bound, we can
conclude a strong converse: for any sequence of quantum-feedback-assisted
codes for a channel ${\mathcal{N}}$ with rate $R>I({\mathcal{N}})$, the
success probability decays exponentially to zero as $n\rightarrow\infty$.
\end{theorem}

Note that the second statement in Theorem~\ref{thm:sc-feedback} has in fact
already been proved in \cite[Section~IV-E1]{BDHSW12}, via the channel
simulation technique. However, our new contribution here is to provide the
bound in (\ref{eq:QFA-strong-conv-exponent}) on the strong converse exponent,
in addition to providing an arguably more direct proof of the theorem. It
remains an open question to determine if the strong converse exponent bound in
(\ref{eq:QFA-strong-conv-exponent}) is optimal (i.e., if there exists a
quantum-feedback-assisted communication scheme achieving this exponent in the
strong converse regime).

\section{R\'enyi relative entropies}
\label{notation}

For two Hilbert spaces $\hil,\kil$, let $\B(\hil,\kil)$ denote the set of bounded linear operators from 
$\hil$ to $\kil$. When $\kil=\hil$, we use the shorthand notation $\B(\hil)$.
We restrict ourselves to finite-dimensional Hilbert spaces throughout this
paper. The Schatten $\alpha$-norm of an operator $X$ is defined as%
\begin{equation}
\left\Vert X\right\Vert _{\alpha}\equiv\Tr\{(\sqrt{X^{*}X})^{\alpha
}\}^{1/\alpha},
\end{equation}
for $\alpha\geq1$. Let $\mathcal{B}\left(  \mathcal{H}\right)  _{+}$ denote
the subset of positive semi-definite operators; we often simply say that an
operator is \textquotedblleft positive\textquotedblright\ if it is positive
semi-definite. We also write $X\geq0$ if $X\in\mathcal{B}\left(
\mathcal{H}\right)  _{+}$. An\ operator $\rho$ is in the set $\mathcal{S}%
\left(  \mathcal{H}\right)  $\ of density operators if $\rho\in\mathcal{B}%
\left(  \mathcal{H}\right)  _{+}$ and $\Tr\left\{  \rho\right\}  =1$. 
We denote by $\B(\hil)_{++}$ and $\S(\hil)_{++}$ the set of positive definite
operators and states on $\hil$, respectively.

The
tensor product of two Hilbert spaces $\mathcal{H}_{A}$ and $\mathcal{H}_{B}$
is denoted by $\mathcal{H}_{A}\otimes\mathcal{H}_{B}$.\ Given a bipartite
density operator $\rho_{AB}\in{\mathcal{S}}(\mathcal{H}_{A}\otimes
\mathcal{H}_{B})$, we write $\rho_{A}=\Tr_{B}\left\{  \rho_{AB}\right\}  $ for
the reduced density operator on system $A$. A linear map $\mathcal{N}%
_{A\rightarrow B}:\mathcal{B}\left(  \mathcal{H}_{A}\right)  \rightarrow
\mathcal{B}\left(  \mathcal{H}_{B}\right)  $\ is positive if $\mathcal{N}%
_{A\rightarrow B}\left(  \sigma_{A}\right)  \in\mathcal{B}\left(
\mathcal{H}_{B}\right)  _{+}$ whenever $\sigma_{A}\in\mathcal{B}\left(
\mathcal{H}_{A}\right)  _{+}$. Let id$_{A}$ denote the identity map acting on
a system $A$. A linear map $\mathcal{N}_{A\rightarrow B}$ is completely
positive if the map id$_{R}\otimes\mathcal{N}_{A\rightarrow B}$ is positive
for a reference system $R$ of arbitrary size. A linear map $\mathcal{N}%
_{A\rightarrow B}$ is trace-preserving if $\Tr\left\{  \mathcal{N}%
_{A\rightarrow B}\left(  \tau_{A}\right)  \right\}  =\Tr\left\{  \tau
_{A}\right\}  $ for all input operators $\tau_{A}\in\mathcal{B}\left(
\mathcal{H}_{A}\right)  $. If a linear map is completely positive and
trace-preserving (CPTP), we say that it is a quantum channel or quantum
operation. A positive operator-valued measure (POVM) is a set $\left\{
\Lambda^{m}\right\}  $ of positive operators such that $\sum_{m}\Lambda^{m}=I$.

The quantum R\'{e}nyi relative entropy of order $\alpha\in\lbrack
0,1)\cup(1,\infty)$ between two non-zero positive semidefinite operators $\rho$ and $\sigma$ is given
by \cite{P86}%
\begin{equation}
D_{\alpha}(\rho\Vert\sigma)\equiv\left\{
\begin{array}
[c]{cc}%
\frac{1}{\alpha-1}\log \frac{1}{\Tr\rho}\text{Tr}\left\{  \rho^{\alpha}\sigma^{1-\alpha}\right\}
& \text{if }\rho\not \perp \sigma\text{ and (supp}\left(  \rho\right)
\subseteq\text{supp}\left(  \sigma\right)  \text{ or }\alpha\in\lbrack
0,1)\text{\ )}\\
+\infty & \text{otherwise}.
\end{array}
\right.  , \label{eq:Renyi-rel-ent}%
\end{equation}
with the support conditions established in \cite{TCR09}.
Here and henceforth we use the convention that powers of a positive
semidefinite operator $X$ are taken only on its support, i.e., if
$x_{1},\ldots,x_{r}$ are the strictly positive eigenvalues of $X$ with
corresponding spectral projections $P_{1},\ldots,P_{r}$, then $X^{t}\equiv
\sum_{i=1}^{r} x_{i}^{t} P_{i}$ for every $t\in\mathbb{R}$. In particular,
$X^{0}$ denotes the projection onto the support of $X$.

Recently, the sandwiched R\'{e}nyi relative entropy \cite{MDSFT13,WWY13} was
introduced. It is defined for $\alpha\in(0,1)\cup(1,\infty)$ as follows:%
\begin{equation}
\widetilde{D}_{\alpha}(  \rho\Vert\sigma)  \equiv\left\{
\begin{array}
[c]{cc}%
\frac{1}{\alpha-1}\log\left[ \frac{1}{\Tr\rho} \text{Tr}\left\{  \left(  \sigma^{\left(
1-\alpha\right)  /2\alpha}\rho\sigma^{\left(  1-\alpha\right)  /2\alpha
}\right)  ^{\alpha}\right\}  \right]  &
\begin{array}
[c]{c}%
\text{if }\rho\not \perp \sigma\text{ and (supp}\left(  \rho\right)
\subseteq\text{supp}\left(  \sigma\right) \\
\text{or }\alpha\in(0,1)\text{ )}%
\end{array}
\\
+\infty & \text{otherwise}%
\end{array}
\right.  . \label{eq:def-sandwiched}%
\end{equation}

It is known \cite{MH,MO} that for any fixed $\rho,\sigma$,
\begin{align}\label{alpha mon}
\alpha\mapsto D_{\alpha}(\rho\|\sigma)\ds\ds\text{and}\ds\ds\
\alpha\mapsto \wt D_{\alpha}(\rho\|\sigma)\ds\ds\text{are monotone increasing},
\end{align}
and in the limit $\alpha\to 1$, they both give the relative entropy \cite{MDSFT13, WWY13}:
\begin{equation}\label{Renyi limit}
\lim_{\alpha\rightarrow1}\widetilde{D}_{\alpha}(\rho\Vert\sigma)=\lim
_{\alpha\rightarrow1}D_{\alpha}(\rho\Vert\sigma)=D(\rho\Vert\sigma)\equiv D_1(\rho\Vert\sigma).
\end{equation}

The R\'enyi relative entropies have several desirable properties which justify viewing them
as distinguishability measures. In particular, $\widetilde{D}_{\alpha}\left(
\rho\Vert\sigma\right)$ satisfies the following data-processing inequality
for $\alpha\in\lbrack1/2,1)\cup(1,\infty)$ \cite{FL13,B13monotone,MDSFT13, WWY13,MO}:
\begin{equation}
\widetilde{D}_{\alpha}(  \rho\Vert\sigma)  \geq\widetilde
{D}_{\alpha}\left(  \mathcal{N}\left(  \rho\right)  \Vert\mathcal{N}\left(
\sigma\right)  \right)  ,
\end{equation}
where $\mathcal{N}$ is a CPTP\ map. A similar inequality holds for $D_{\alpha
}(  \rho\Vert\sigma)  $ when $\alpha\in\lbrack0,1)\cup(1,2]$
\cite{P86}.

The following simple lemma relates the hypothesis testing relative entropy to
the sandwiched R\'{e}nyi relative entropy. The idea for its proof goes back to
\cite{HP,N,ON}.

\begin{lemma}\label{lem:sandwich-to-htre}
Let $\rho,\sigma\in\mathcal{S}\left(\mathcal{H}\right)$ be such that $\supp\rho\subseteq\supp\sigma$. For any $Q\in\B(\hil)$ such that $0\le Q\le I$, and 
any $\alpha>1$, 
\begin{align}\label{Nagaoka ineq}
-\log\Tr Q\sigma \le \widetilde D_{\alpha}(\rho\|\sigma)- \frac{\alpha}{\alpha-1}\log\Tr Q\rho.
\end{align}
In particular, for any $\alpha>1$ and any $\varepsilon\in\left(0,1\right)$,
\begin{equation}
D_{H}^{\varepsilon}(  \rho\Vert\sigma)  \leq\widetilde{D}_{\alpha
}(  \rho\Vert\sigma)  +\frac{\alpha}{\alpha-1}\log\left(  \frac
{1}{1-\varepsilon}\right)  . \label{eq:sandwich-to-htre}%
\end{equation}
\end{lemma}

\begin{proof}
Let $p    \equiv\text{Tr}\left\{  Q\rho\right\}$ and $q  \equiv\text{Tr}\left\{  Q\sigma\right\}$.
By the monotonicity of the sandwiched R\'{e}nyi relative entropy for
$\alpha>1$, we find that%
\begin{align}
\widetilde{D}_{\alpha}(  \rho\Vert\sigma)   &  \geq\widetilde
{D}_{\alpha}\left(  \left(  p,1-p\right)  \Vert\left(  q,1-q\right)  \right)
\label{eq:sandwich-to-p-and-q-1}\\
&  =\frac{1}{\alpha-1}\log\left[  p^{\alpha}q^{1-\alpha}+\left(  1-p\right)
^{\alpha}\left(  1-q\right)  ^{1-\alpha}\right] \\
&  \geq\frac{1}{\alpha-1}\log\left[  p^{\alpha}q^{1-\alpha}\right] \\
&  =\frac{\alpha}{\alpha-1}\log p-\log q,
\label{eq:sandwich-to-p-and-q-last}%
\end{align}
from which \eqref{Nagaoka ineq} follows.
The statement in 
\eqref{eq:sandwich-to-htre} follows by
optimizing over all $Q$ such that $\text{Tr}\left\{  Q\rho\right\}  \geq1-\varepsilon$.
\end{proof}
\bigskip

Recall the definition of the channel R\'enyi relative entropies
in \eqref{channel Renyi1}--\eqref{channel Renyi2}.
Let $A'$ be a copy of $A$, let $e_1,\ldots,e_d$ be an orthonormal basis in $A$, and define
$\ket{\Gamma_{A'A}}\equiv\ket{\sum_{i=1}^d e_i\otimes e_i}$, 
$\Gamma_{A'A}\equiv\pr{\Gamma_{A'A}}$. Then we have the following:

\begin{lemma}\label{lemma:channel Renyi0}
Let $A'$ be a copy of $A$.
For any system $R$ and any pure state $\psi_{RA}$, there exists a state $\rho_{A'}$ on $A'$ such that for any two channels
$\N_1,\N_2$ from $A$ to some system $B$, and any $\alpha>0$,
we have 
\begin{align}
\widetilde D_{\alpha}({\mathcal{N}}_{1}\Vert{\mathcal{N}}_{2})
&=
\widetilde D_{\alpha}\!\bz\N_1\bz\rho_{A'}^{1/2}\pr{\Gamma_{A'A}}\rho_{A'}^{1/2}\jz
\Big\|\N_2\bz\rho_{A'}^{1/2}\pr{\Gamma_{A'A}}\rho_{A'}^{1/2}\jz\jz\label{channel Renyi6}\\
&=
\widetilde D_{\alpha}\!\bz\rho_{A'}^{1/2}\N_1(\Gamma_{A'A})\rho_{A'}^{1/2}\Big\|\rho_{A'}^{1/2}\N_2(\Gamma_{A'A})\rho_{A'}^{1/2}\jz.\label{channel Renyi7}
\end{align}
Moreover, the same identities hold for $D_{\alpha}$.
\end{lemma}

\noindent We give a proof of Lemma~\ref{lemma:channel Renyi0} in Appendix~\ref{app:CB}.

\begin{lemma}\label{lemma:channel Renyi}
Let $\N_1,\N_2$ be quantum channels from system $A$ to system $B$.
For every $\alpha\in[1/2,+\infty)$, the channel R\'enyi relative entropies can be written as
\begin{align}
\widetilde D_{\alpha}({\mathcal{N}}_{1}\Vert{\mathcal{N}}_{2})  
& =
\sup\left\{\widetilde D_{\alpha}({\mathcal{N}}_{1}(\psi_{RA})\|{\mathcal{N}}_{2}(\psi_{RA})):\,\psi_{RA}\s\text{state on $RA$, where $R$ is arbitrary}\ds\right\}\label{channel Renyi3}\\
&=
\sup\left\{\widetilde D_{\alpha}({\mathcal{N}}_{1}(\psi_{RA})\|{\mathcal{N}}_{2}(\psi_{RA})):\,\psi_{RA}\s\text{pure state on $RA$, where $R\cong A$}\ds\right\}\label{channel Renyi4}\\
&=
\sup\left\{\widetilde D_{\alpha}\!\bz\rho_{A'}^{1/2}\N_1(\Gamma_{A'A})\rho_{A'}^{1/2}\big\|\rho_{A'}^{1/2}\N_2(\Gamma_{A'A})\rho_{A'}^{1/2}\jz:\,\rho_{A'}\s\text{state on $A'$, where $A'\cong A$}\ds\right\}.\label{channel Renyi5}
\end{align}
Analogous formulas hold for $D_{\alpha}({\mathcal{N}}_{1}\Vert{\mathcal{N}}_{2})$
in \eqref{channel Renyi1} and $\alpha\in[0,2]$.
\end{lemma}
\begin{proof}
According to \cite{FL13}, $\wt D_{\alpha}$ is jointly quasi-convex for $\alpha\in[1/2,+\infty)$, 
and by \cite{Ando,Lieb-convexity,P86}, the same holds for $D_{\alpha}$ and $\alpha\in[0,2]$. Hence, the optimizations in 
 \eqref{channel Renyi1}--\eqref{channel Renyi2} can be restricted to pure states, and the rest of the proof 
 is immediate from Lemma \ref{lemma:channel Renyi0}.
\end{proof}
\medskip

When the second channel is a replacer channel, the sandwiched channel R\'enyi relative entropy has a special representation as explained below. 
This will be key to our approach of obtaining strong converse bounds.

A quantum channel $\N_{A\to B}$ induces a map from $L_{1}(\mathcal{B}({\mathcal{H}}%
_{A}))\rightarrow L_{\alpha}(\mathcal{B}({\mathcal{H}}_{B}))$, where
$L_{\alpha}(\mathcal{B}({\mathcal{H}})))$ denotes the space $\mathcal{B}%
({\mathcal{H}})$ together with the Schatten $\alpha$-norm $\Vert
X\Vert_{\alpha}$. The space $L_{\alpha}(\mathcal{B}({\mathcal{H}}))$ has a
canonical operator space structure \cite{P}, a certain sequence of norms on
the spaces $M_{n}(L_{\alpha}(\mathcal{B}({\mathcal{H}})))$:
\begin{equation}
\Vert Y\Vert_{M_{n}(L_{\alpha}(\mathcal{B}({\mathcal{H}})))}\equiv\sup_{A,B\in
M_{n}}\frac{\Vert(A\otimes I_{{\mathcal{H}}})Y(B\otimes I_{{\mathcal{H}}%
})\Vert_{\alpha}}{\Vert A\Vert_{2\alpha}\Vert B\Vert_{2\alpha}}.
\end{equation}
One can then define the completely bounded $(1\rightarrow\alpha)$-norm of
${\mathcal{N}}:L_{1}(\mathcal{B}({\mathcal{H}}_{A}))\rightarrow L_{\alpha
}(\mathcal{B}({\mathcal{H}}_{B}))$ as
\begin{equation}
\sup_{n}\left\Vert \id_{n}\otimes{\mathcal{N}}\right\Vert _{1\rightarrow
\alpha}\equiv\sup_{n}\sup_{Y}\frac{\Vert(\id_{n}\otimes{\mathcal{N}}%
)(Y)\Vert_{M_{n}(L_{\alpha}(\mathcal{B}({\mathcal{H}}_{B})))}}{\Vert
Y\Vert_{M_{n}(L_{1}(\mathcal{B}({\mathcal{H}}_{A})))}}.
\end{equation}
For our purposes, it will be more useful to write the completely bounded
$(1\rightarrow\alpha)$-norm of a quantum channel ${\mathcal{N}}$ as%
\begin{equation}
\Vert{\mathcal{N}}\Vert_{\text{CB},1\rightarrow\alpha}
=
\sup_{X\in \B(\hil_{A'}\otimes\hil_A)_+}\frac{\Vert(\id\otimes{\mathcal{N}%
})(X)\Vert_{\alpha}}{\Vert\Tr_{A}\{X\}\Vert_{\alpha}}
=
\sup_{|\psi\rangle
\in \hil_{A'}\otimes \hil_A}\frac{\Vert(\id\otimes{\mathcal{N}}%
)(|\psi\rangle\langle\psi|)\Vert_{\alpha}}{\Vert\Tr_{A}\{|\psi\rangle
\langle\psi|\}\Vert_{\alpha}}, \label{cb1toalpha}%
\end{equation}
where $A'$ is any system with dimension at least that of $A$; in particular, $A'$ can be taken to be a fixed copy of $A$.
This follows from \cite{DJKR} and
Eq.~(8) of \cite{Jencova}, where these norms have already been considered in
the context of quantum information theory. 
The above representation of the completely bounded
$(1\rightarrow\alpha)$-norm will
prove useful later due to the following connection between the sandwiched
R\'{e}nyi relative entropy and the Schatten $\alpha$-norm:%
\begin{equation}
\widetilde{D}_{\alpha}(  \rho\Vert\sigma)  =\frac{\alpha}{\alpha
-1}\log\left\Vert \sigma^{\frac{1-\alpha}{2\alpha}}\rho\sigma^{\frac{1-\alpha
}{2\alpha}}\right\Vert _{\alpha}, \label{entropytonorm}%
\end{equation}
where $\rho$ and $\sigma$ are density operators. Throughout, $\Theta_{X}$
denotes the map
\begin{equation}
\Theta_{X}(Y)\equiv X^{1/2}YX^{1/2},
\end{equation}
where $X$ is a positive operator. 

The following lemma is from \cite{GW}; for readers' convenience, we give a detailed proof in Appendix~\ref{app:CB}.

\begin{lemma}\label{lemma:CB}
Let $\N=\N_{A\to B}$ be a quantum channel and $\R_{\sigma_B}(\cdot)\equiv\Tr\{\cdot\}\sigma_B$ be a replacer channel with some fixed state $\sigma_B$.
For every $\alpha\in(1,+\infty)$,
\begin{equation}
\wt D_{\alpha}(\N\|\R_{\sigma_B})
=
\sup_{\psi_{RA}}\widetilde{D}_{\alpha}\left(  \mathcal{N}_{A\rightarrow
B}(\psi_{RA})\Vert\psi_{R}\otimes\sigma_{B}\right)  
=
\frac{\alpha}{\alpha
-1}\log\left\Vert \Theta_{\sigma_{B}^{\frac{1-\alpha}{\alpha}}}\circ
{\mathcal{N}}\right\Vert _{\operatorname{CB},1\rightarrow\alpha}.%
\label{lemmanorm}%
\end{equation}
\end{lemma}
\medskip

We will also use the following R\'enyi mutual information quantities, originally defined in \cite{WWY13,B13monotone,GW}. For every bipartite state $\rho_{RB}$, and every $\alpha\in(0,+\infty)$, let 
\begin{align}
I_{\alpha}(R;B)_\rho&\equiv\inf_{\sigma_B\in\S(\hil_B)}D_{\alpha}(\rho_{RB}\|\rho_R\otimes\sigma_B),\\
\wt I_{\alpha}(R;B)_\rho&\equiv\inf_{\sigma_B\in\S(\hil_B)}\wt D_{\alpha}(\rho_{RB}\|\rho_R\otimes\sigma_B).
\end{align}
These quantities appeared in the direct and strong converse exponents of \cite{HT}. We also define the channel R\'enyi mutual informations. For any CPTP map
$\N_{A\to B}:\,\B(\hil_A)\to\B(\hil_B)$, let
\begin{align}
I_{\alpha}(\N)&\equiv\sup_{\psi_{RA}\in\S(\hil_{RA})}I_{\alpha}(R;B)_{\omega},\label{channel div}\\
\wt I_{\alpha}(\N)&\equiv\sup_{\psi_{RA}\in\S(\hil_{RA})}\wt I_{\alpha}(R;B)_{\omega}\label{EA4},
\end{align}
where $\omega_{RB} \equiv \N_{A\to B}(\psi_{RA})$.

\begin{lemma}\label{lemma:channel div repr}
Let $A'$ be a copy of $A$. Then 
\begin{align}
I_{\alpha}(\N)&=\sup_{\rho_{R}\in\S(\hil_{R})}I_{\alpha}(R;B)_\omega,
\ds\ds\ds \alpha\in[0,2],\\
\wt I_{\alpha}(\N)&=\sup_{\rho_{R}\in\S(\hil_{R})}I_{\alpha}(R;B)_\omega,
\ds\ds\ds \alpha\in [1/2,+\infty),
\end{align}
where $\omega_{RB} \equiv \N_{A\to B}\!\left({\rho_{R}^{1/2}\pr{\Gamma_{RA}}\rho_{R}^{1/2}}\right)$.
\end{lemma}
\begin{proof}
According to \cite{FL13}, $\wt D_{\alpha}$ is monotone non-increasing under partial trace for $\alpha\in[1/2,+\infty)$, 
and by \cite{P86}, the same holds for $D_{\alpha}$ and $\alpha\in[0,2]$. Hence, by taking purifications of 
$\psi_{RA}$ in \eqref{channel div} and \eqref{EA4}, the values can only increase. Thus, the optimizations in 
\eqref{channel div} and \eqref{EA4} can be restricted to pure states. Using Lemma \ref{lemma:channel Renyi0}
with $\N_1=\N$ and $\N_2=\R_{\sigma_B}$, the assertions follow.
\end{proof}
\medskip


Note that for $\alpha=1$, the above quantities are defined using the relative entropy $D=D_1$, and we have 
$I_{1}(R;B)_\rho=\wt I_{1}(R;B)_\rho\equiv  I(R;B)_\rho=D(\rho_{RB}\| \rho_R \otimes \rho_B)$,
and $I_{1}(\N)=\wt I_{1}(\N)=I(\N)$, where $I(\N)$ is defined in \eqref{IN}.
We will need the following extensions of \eqref{alpha mon}--\eqref{Renyi limit}:

\begin{lemma}\label{MI limits}
\begin{enumerate}
\item\label{channel div limit}
For any two channels $\N_1,\,\N_2$, $D_{\alpha}(\N_1\|\N_2)$ and $\wt D_{\alpha}(\N_1\|\N_2)$ are monotone increasing in $\alpha$, and
\begin{align}\label{eq:channel div limit}
\lim_{\alpha\to 1}\wt D_{\alpha}(\N_1\|\N_2)=
\lim_{\alpha\to 1} D_{\alpha}(\N_1\|\N_2)=
D(\N_1\|\N_2).
\end{align}

\item\label{mi limit}
For every bipartite state $\rho_{RB}$, $I_{\alpha}(R;B)_\rho$ and $\wt I_{\alpha}(R;B)_\rho$ are monotone increasing in $\alpha$, and 
\begin{align}
\lim_{\alpha\to 1}I_{\alpha}(R;B)_\rho=
\lim_{\alpha\to 1}\wt I_{\alpha}(R;B)_\rho=
I(R;B)_\rho.
\end{align}

\item\label{channel info limit}
For every channel $\N$, $I_{\alpha}(\N)$ and $\wt I_{\alpha}(\N)$ are monotone increasing in $\alpha$, and 
\begin{align}
\lim_{\alpha\to 1}I_{\alpha}(\N)=
\lim_{\alpha\to 1}\wt I_{\alpha}(\N)=
I(\N).
\end{align}
\end{enumerate}
\end{lemma}
\begin{proof}
See Appendix~\ref{app:CB}.
\end{proof}
\medskip

The channel R\'enyi mutual informations also have the following geometric interpretation, as 
the ``distance'' of the channel from the set of all replacer channels, where the ``distance'' is 
measured by the channel R\'enyi divergences. 
See Section \ref{sec:extension} for the relevance of this geometric picture.
\begin{lemma}\label{lemma:geom int}
For every channel $\N_{A\to B}$, and every $\alpha\in[1/2,+\infty)$,
\begin{align}
\wt I_{\alpha}(\N)
&=
\inf_{\sigma_B\in\S(\hil_B)}\wt D_{\alpha}(\N\|\R_{\sigma_B}) \label{EA-1}.
\end{align}
\end{lemma}
\begin{proof}
See Appendix~\ref{app:CB}.
\end{proof}

\section{The strong converse theorem for adaptive quantum channel
discrimination}

\subsection{Quantum Stein's lemma for adaptive channel discrimination}

\label{sec:Stein proof}

This section provides a proof of Theorem~\ref{stein}. In the setting of this
theorem, we seek to distinguish between an arbitrary quantum channel
${\mathcal{N}}$ and a \textquotedblleft replacer\textquotedblright\ channel
${\mathcal{R}}$ that maps all states $\omega_A$ to a fixed state $\sigma$, i.e.,
${\mathcal{R}}(\omega_{A})=\Tr\{\omega_{A}\}\sigma_{B}$. We allow the
preparation of an arbitrary input state $\rho_{R_{1}A_{1}}=\tau_{R_{1}A_{1}}$,
where $R_{1}$ is an ancillary register. The $i$th use of a channel
accepts the register $A_{i}$ as input and produces the register $B_{i}$ as
output. After each invocation of the channel, an adaptive operation
$\mathcal{A}^{(i)}$ is applied to the registers $R_{i}$ and $B_{i}$, yielding
a quantum state $\rho_{R_{i+1}A_{i+1}}$ or $\tau_{R_{i+1}A_{i+1}}$ in
registers $R_{i+1}A_{i+1}$, depending on whether the channel is equal to
${\mathcal{N}}$ or ${\mathcal{R}}$. That is,
\begin{align}
\rho_{R_{i+1}A_{i+1}} & \equiv\mathcal{A}_{R_{i}B_{i}\rightarrow
R_{i+1}A_{i+1}}^{\left( i\right)  }( \rho_{R_{i}B_{i}}) , &  &
\rho_{R_{i}B_{i}}\equiv\mathcal{N}_{A_{i}\rightarrow B_{i}}( \rho
_{R_{i}A_{i}}) \label{eq:rho-adaptive}\\
\tau_{R_{i+1}A_{i+1}} & \equiv\mathcal{A}_{R_{i}B_{i}\rightarrow
R_{i+1}A_{i+1}}^{\left( i\right)  }( \tau_{R_{i}B_{i}}) , &  &
\tau_{R_{i}B_{i}}\equiv\mathcal{R}_{A_{i}\rightarrow B_{i}}( \tau
_{R_{i}A_{i}}) \label{eq:tau-adaptive}%
\end{align}
for every $1\le i<n$ on the left-hand side, and for every $1\le i\le n$ on the
right-hand side. Finally, a quantum measurement $\{Q_{R_{n}B_{n}}%
,I_{R_{n}B_{n}}-Q_{R_{n}B_{n}}\}$ is performed on the systems $R_{n}B_{n}$ to
decide which channel was applied. Such a general protocol is depicted in
Figures \ref{FigureChannelDiscrim1} and \ref{FigureChannelDiscrim2}. Note that
since $\mathcal{R}$ is a replacer channel, we can write%
\begin{equation}
\tau_{R_{i}B_{i}}=\tau_{R_{i}}\otimes\sigma_{B_{i}}%
,\mbox{ }\mbox{ }\mbox{ }\mbox{ }\mbox{ }\mbox{ } 1\le i\le n.
\end{equation}

Recall the hypothesis testing relative entropy $D_{H}^{\varepsilon}(\rho\Vert\sigma)  $ from (\ref{eq:htre}) and the \textquotedblleft
adaptive hypothesis testing relative entropy\textquotedblright$D_{H,\ad}%
^{\varepsilon}(  \mathcal{N}^{\otimes n}\Vert\mathcal{R}^{\otimes
n})  $ from \eqref{eq:adaptive-htre}. So $D_{H}^{\varepsilon}(
\mathcal{N}_{A_{n}\rightarrow B_{n}}(  \rho_{R_{n}A_{n}})
\Vert\tau_{R_{n}}\otimes\sigma_{B_{n}})  $ denotes the hypothesis
testing relative entropy in which there is a fixed initial state $\rho
_{R_{1}A_{1}}$ and fixed adaptive maps $\{ \mathcal{A}_{R_{i}B_{i}\rightarrow
R_{i+1}A_{i+1}}^{\left(  i\right)  }\} _{i\in\left\{  1,\ldots,n-1\right\}  }$.

Clearly, we have that%
\begin{equation}
\liminf_{n\rightarrow\infty}\frac{1}{n}D_{H,\ad}^{\varepsilon}(
\mathcal{N}^{\otimes n}\Vert\mathcal{R}^{\otimes n})  \geq
\sup_{\psi_{RA}}D(  \mathcal{N}_{A\rightarrow B}(  \psi_{RA})
\Vert\psi_{R}\otimes\sigma_{B})  ,
\label{eq:EA-stein-lower-bound}%
\end{equation}
by employing a tensor-power strategy with no adaptation (i.e., we can simply
invoke the direct part of the usual quantum Stein's lemma). In more detail,
the initial state of this strategy is the optimal $\psi_{RA}$ in
(\ref{eq:EA-stein-lower-bound}) and each map $\mathcal{A}_{R_{i}%
B_{i}\rightarrow R_{i+1}A_{i+1}}^{\left(  i\right)  }$ simply prepares the
state $\psi_{RA}$ at the input of the $i+1$st channel while acting
as the identity map on the $i$ states $(\mathcal{N}_{A\rightarrow B}(
\psi_{RA})  )^{\otimes i}$ or $(\mathcal{R}_{A\rightarrow B}(
\psi_{RA})  )^{\otimes i}$ (so the strategy is non-adaptive). After the
$n$th channel has acted, the discriminator performs a binary
collective measurement on the state $(\mathcal{N}_{A\rightarrow B}(
\psi_{RA})  )^{\otimes n}$ or $(\mathcal{R}_{A\rightarrow B}(
\psi_{RA})  )^{\otimes n}$ to decide which channel was applied. So the
lower bound in (\ref{eq:EA-stein-lower-bound}) follows directly from the state
discrimination result in~(\ref{eq:steins-htre}).

The more interesting part is to show that this strategy is asymptotically
optimal, i.e., that%
\begin{equation}
\limsup_{n\rightarrow\infty}\frac{1}{n}D_{H,\ad}^{\varepsilon}(
\mathcal{N}^{\otimes n}\Vert\mathcal{R}^{\otimes n})  \leq
\sup_{\psi_{RA}}D(  \mathcal{N}_{A\rightarrow B}(  \psi_{RA})
\Vert\psi_{R}\otimes\sigma_{B})  . \label{eq:EA-steins}%
\end{equation}
Since this inequality is trivial when $\sup_{\psi_{RA}}D(  \mathcal{N}_{A\rightarrow B}(  \psi_{RA})
\Vert\psi_{R}\otimes\sigma_{B})=+\infty$, we assume the contrary for the rest.
We start by bounding the adaptive hypothesis testing relative entropy in terms
of the sandwiched R\'{e}nyi relative entropy.

Throughout this section, the parameter $\alpha$ is assumed to be strictly larger than
one and we fix some constant $\varepsilon\in(0,1)$. We fix some input state
$\rho_{R_{1}A_{1}}$ and an adaptive strategy $(\mathcal{A}^{(1)}%
,\cdots,\mathcal{A}^{(n-1)})$. Lemma~\ref{lem:sandwich-to-htre} implies that%
\begin{multline}
D_{H}^{\varepsilon}(  \mathcal{N}_{A_{n}\rightarrow B_{n}}(
\rho_{R_{n}A_{n}})  \Vert\tau_{R_{n}}\otimes\sigma_{B_{n}})
\leq
\widetilde{D}_{\alpha}(  \mathcal{N}_{A_{n}\rightarrow B_{n}}(
\rho_{R_{n}A_{n}})  \Vert\tau_{R_{n}}\otimes\sigma_{B_{n}}) \\
+\frac{\alpha}{\alpha-1}\log\left(  \frac{1}{1-\varepsilon}\right)  .
\end{multline}

We now focus on the $\widetilde{D}_{\alpha}$\ term. Let $\Theta_{\omega}$
denote the completely positive map $\Theta_{\omega}(  X)
=\omega^{1/2}X\omega^{1/2}$ that conjugates $X$ by a positive operator
$\omega^{1/2}$. From (\ref{entropytonorm}), it follows that%
\begin{align}
&  \widetilde{D}_{\alpha}(  \mathcal{N}_{A_{n}\rightarrow B_{n}}(
\rho_{R_{n}A_{n}})  \Vert\tau_{R_{n}}\otimes\sigma_{B_{n}})
\nonumber\\
&  =\frac{\alpha}{\alpha-1}\log\left\Vert \left(  \tau_{R_{n}}\otimes
\sigma_{B_{n}}\right)  ^{\frac{1-\alpha}{2\alpha}}\mathcal{N}_{A_{n}%
\rightarrow B_{n}}(  \rho_{R_{n}A_{n}})  \left(  \tau_{R_{n}%
}\otimes\sigma_{B_{n}}\right)  ^{\frac{1-\alpha}{2\alpha}}\right\Vert
_{\alpha}\\
&  =\frac{\alpha}{\alpha-1}\log\left\Vert \left(  \Theta_{\sigma_{B_{n}%
}^{\frac{1-\alpha}{\alpha}}}\circ\mathcal{N}_{A_{n}\rightarrow B_{n}}\right)
\left(  \tau_{R_{n}}^{\frac{1-\alpha}{2\alpha}}\rho_{R_{n}A_{n}}\tau_{R_{n}%
}^{\frac{1-\alpha}{2\alpha}}\right)  \right\Vert _{\alpha}.%
\end{align}
Let us focus on the expression inside the logarithm:%
\begin{align}
&  \left\Vert \left(  \Theta_{\sigma_{B_{n}}^{\frac{1-\alpha}{\alpha}}}%
\circ\mathcal{N}_{A_{n}\rightarrow B_{n}}\right)  \left(  \tau_{R_{n}}%
^{\frac{1-\alpha}{2\alpha}}\rho_{R_{n}A_{n}}\tau_{R_{n}}^{\frac{1-\alpha
}{2\alpha}}\right)  \right\Vert _{\alpha}\nonumber\\
&  =\frac{\left\Vert \left(  \Theta_{\sigma_{B_{n}}^{\frac{1-\alpha}{\alpha}}%
}\circ\mathcal{N}_{A_{n}\rightarrow B_{n}}\right)  \left(  \tau_{R_{n}}%
^{\frac{1-\alpha}{2\alpha}}\rho_{R_{n}A_{n}}\tau_{R_{n}}^{\frac{1-\alpha
}{2\alpha}}\right)  \right\Vert _{\alpha}}{\left\Vert \tau_{R_{n}}%
^{\frac{1-\alpha}{2\alpha}}\rho_{R_{n}}\tau_{R_{n}}^{\frac{1-\alpha}{2\alpha}%
}\right\Vert _{\alpha}}\cdot\left\Vert \tau_{R_{n}}^{\frac{1-\alpha}{2\alpha}%
}\rho_{R_{n}}\tau_{R_{n}}^{\frac{1-\alpha}{2\alpha}}\right\Vert _{\alpha}
\end{align}
\begin{align}
&  \leq\left(  \sup_{X_{R_{n}A_{n}}\geq0}\frac{\left\Vert \left(  \Theta
_{\sigma_{B_{n}}^{\frac{1-\alpha}{\alpha}}}\circ\mathcal{N}_{A_{n}\rightarrow
B_{n}}\right)  \left(  X_{R_{n}A_{n}}\right)  \right\Vert _{\alpha}%
}{\left\Vert X_{R_{n}}\right\Vert _{\alpha}}\right)  \cdot\left\Vert
\tau_{R_{n}}^{\frac{1-\alpha}{2\alpha}}\rho_{R_{n}}\tau_{R_{n}}^{\frac
{1-\alpha}{2\alpha}}\right\Vert _{\alpha}\\
&  =\left\Vert \Theta_{\sigma_{B}^{\frac{1-\alpha}{\alpha}}}\circ
\mathcal{N}\right\Vert _{\text{CB,}1\rightarrow\alpha}\cdot\left\Vert
\tau_{R_{n}}^{\frac{1-\alpha}{2\alpha}}\rho_{R_{n}}\tau_{R_{n}}^{\frac
{1-\alpha}{2\alpha}}\right\Vert _{\alpha}.\label{changetonorm}
\end{align}
The equality in (\ref{changetonorm}) follows from the characterisation of the
completely bounded $(1\rightarrow\alpha)$-norm given in (\ref{cb1toalpha}).
Rewriting this inequality in terms of the sandwiched R\'{e}nyi relative
entropy, we have that%
\begin{align}
&  \widetilde{D}_{\alpha}(  \mathcal{N}_{A_{n}\rightarrow B_{n}}(
\rho_{R_{n}A_{n}})  \Vert\tau_{R_{n}}\otimes\sigma_{B_{n}})
\nonumber\\
&  \leq\frac{\alpha}{\alpha-1}\log\left\Vert \Theta_{\sigma_{B}^{\frac
{1-\alpha}{\alpha}}}\circ\mathcal{N}\right\Vert _{\text{CB,}1\rightarrow
\alpha}+\frac{\alpha}{\alpha-1}\log\left\Vert \tau_{R_{n}}^{\frac{1-\alpha
}{2\alpha}}\rho_{R_{n}}\tau_{R_{n}}^{\frac{1-\alpha}{2\alpha}}\right\Vert
_{\alpha}\\
&  =\frac{\alpha}{\alpha-1}\log\left\Vert \Theta_{\sigma_{B}^{\frac{1-\alpha
}{\alpha}}}\circ\mathcal{N}\right\Vert _{\text{CB,}1\rightarrow\alpha
}+\widetilde{D}_{\alpha}(  \rho_{R_{n}}\Vert\tau_{R_{n}}) \\
&  \leq\frac{\alpha}{\alpha-1}\log\left\Vert \Theta_{\sigma_{B}^{\frac
{1-\alpha}{\alpha}}}\circ\mathcal{N}\right\Vert _{\text{CB,}1\rightarrow
\alpha}+\widetilde{D}_{\alpha}(  \mathcal{N}_{A_{n-1}\rightarrow B_{n-1}%
}(  \rho_{R_{n-1}A_{n-1}})  \Vert\tau_{R_{n-1}}\otimes
\sigma_{B_{n-1}}) ,
\end{align}
where the last inequality follows from monotonicity of the sandwiched
R\'{e}nyi relative entropy under the map $\Tr_{A_{n}} \circ\mathcal{A}%
_{R_{n-1}B_{n-1}\rightarrow R_{n}A_{n}}^{\left(  n-1\right)  }$.

Note that we are now left with the quantity
$
\widetilde{D}_{\alpha}(  \mathcal{N}_{A_{n-1}\rightarrow B_{n-1}}(
\rho_{R_{n-1}A_{n-1}})  \Vert\tau_{R_{n-1}}\otimes\sigma
_{B_{n-1}})  ,
$
which corresponds to applying the first $n-1$ rounds of the adaptive
discrimination process. We can thus iterate the above argument through all $n$
steps of the adaptive strategy. Noting that $\rho_{R_{1}}=\tau_{R_{1}}$, and
thus $\widetilde{D}_{\alpha}\left( \rho_{R_{1}}\|\tau_{R_{1}}\right) =0$, we
obtain the bound
\begin{align}
\label{Renyi-CB}
\widetilde{D}_{\alpha}\left(  \mathcal{N}_{A_{n}\rightarrow B_{n}}(  \rho_{R_{n}A_{n}})  \Vert\tau_{R_{n}}\otimes\sigma_{B_{n}}\right)
&\leq 
n\cdot\frac{\alpha}{\alpha-1}\log\left\Vert \Theta_{\sigma_{B}^{\frac{1-\alpha}{\alpha}}}\circ\mathcal{N}\right\Vert _{\text{CB,}1\rightarrow\alpha}\\
&=
n\widetilde D_{\alpha}(\N\|\R),\label{Renyi-CB2}
\end{align}
where \eqref{Renyi-CB2} follows from Lemma \ref{lemma:CB}.
This bound is independent of any particular adaptive strategy used for
discriminating these channels. Thus, we can conclude that%
\begin{equation}
\frac{1}{n}D_{H,\ad}^{\varepsilon}\left(  \mathcal{N}^{\otimes n}\middle\Vert
\mathcal{R}^{\otimes n}\right)  
\leq
\widetilde D_{\alpha}(\N\|\R)
+
\frac{1}{n}\cdot\frac{\alpha}{\alpha-1}%
\log\left(  \frac{1}{1-\varepsilon}\right)  .
\end{equation}
Taking the limsup as $n\rightarrow\infty$, we get the $\varepsilon$-independent
bound%
\begin{equation}
\limsup_{n\rightarrow\infty}\frac{1}{n}D_{H,\ad}^{\varepsilon}\left(
\mathcal{N}^{\otimes n}\middle\Vert\mathcal{R}^{\otimes n}\right)  \
\leq
\widetilde D_{\alpha}(\N\|\R).
\end{equation}
Taking now the infimum over $\alpha>1$, the assertion follows due to Lemma \ref{MI limits}.

\subsection{The strong converse exponent for adaptive channel discrimination}

\label{sec:sc proof}

Having just proven a quantum Stein's lemma for adaptive channel
discrimination, it is then natural to study the trade-off between error
probabilities, when we impose the condition that the Type~II error probability
has exponential decay rate $r$ for%
\begin{equation}
r>\sup_{\psi_{RA}}D(  \mathcal{N}_{A\rightarrow B}(  \psi
_{RA})  \Vert\psi_{R}\otimes\sigma_{B})  .
\end{equation}
One expects the Type~I error to tend to one exponentially quickly. Building on
the above results, we identify the strong converse exponent for the channel
discrimination problem (where, as before, we assume that the alternative hypothesis is a
replacer channel). Our result generalizes the quantum state discrimination
result from \cite[Theorem IV.10]{MO}. The notation is the same as in the
previous section; in particular, $\rho_{R_{n}B_{n}}$ and $\tau_{R_{n}B_{n}}$
are as in (\ref{eq:rho-adaptive}) and (\ref{eq:tau-adaptive}), respectively.
Recall the definitions of $\mathrm{\underline{sc}}(r)$ and $\mathrm{\overline
{sc}}(r)$ from \eqref{sci def}--\eqref{scs def}, and the definition of
$\alpha_{n,r}^{A}$ from \eqref{channel alpha}. We will need the following lemma:

\begin{lemma}\label{lemma:channel Renyi finite}
Let $\N$ be a quantum channel from system $A$ to system $B$, and $\sigma_B\in\S(\hil_B)$.
The following are equivalent:
\begin{enumerate}
\item\label{support1}
For every $k\in\bN$, every system $R$, and every $\psi_{RA^k}\in\S(\hil_{RA^k})$, 
$\supp\N^{\otimes k}(\psi_{RA^k})\subseteq\supp \psi_{R}\otimes\sigma_B^{\otimes k}$.
\item\label{support2}
For every $\rho_A\in\S(\hil_A)$, $\supp\N(\rho_A)\subseteq\supp\sigma_B$.
\item\label{support3}
$\wt D_{\alpha}(\N\|\R_{\sigma_B})<+\infty$ for all $\alpha\ge 1$.
\item\label{support4}
$\wt D_{\alpha}(\N\|\R_{\sigma_B})<+\infty$ for some $\alpha\ge 1$.
\end{enumerate}
\end{lemma}
\begin{proof}
\ref{support1}$\imp$\ref{support2} is trivial (by taking $\hil_R=\bC$ and $k=1$).
By Lemma \ref{lemma:channel Renyi}, 
\begin{align*}
\wt D_{\alpha}(\N\|\R_{\sigma_B})=
\sup_{\rho_{A'}\in\S(\hil_{A'})}\wt D_{\alpha}\!\bz\rho_{A'}^{1/2}\N(\Gamma_{A'A})\rho_{A'}^{1/2}\big\|\rho_{A'}\otimes\sigma_B\jz.
\end{align*}
\ref{support2}$\imp$\ref{support1} because
$\supp\N^{\otimes k}(\psi_{RA^k})
\subseteq \supp \psi_{R} \otimes \bigotimes_{i=1}^k \N(\psi_{A_i}) \subseteq \supp \psi_{R} \otimes \sigma_B^{\otimes k}$, which follows by iterating the general inclusion $\supp \omega_{CD} \subseteq
\supp \omega_C \otimes \omega_D$ (see, e.g., \cite[Appendix~B.4]{Renner2005}) and applying \ref{support2}.
If \ref{support2} is satisfied then $\rho_{A'}\mapsto\wt D_{\alpha}\!\bz\rho_{A'}^{1/2}\N(\Gamma_{A'A})\rho_{A'}^{1/2}\big\|\rho_{A'}\otimes\sigma_B\jz$ is a continuous finite-valued function on the compact set $\S(\hil_{A'})$, and hence its supremum is finite, proving \ref{support3}. The implication \ref{support3}$\imp$\ref{support4} is trivial.
Finally, \ref{support4}$\imp$\ref{support2} by applying the definition of $\wt D_\alpha$.
\end{proof}

%

\bigskip

\begin{proof}
[Proof of Theorem~\ref{scrtheorem}]
The statement is empty when $D(\N\|\R)=+\infty$, and hence for the rest we assume the contrary. 

We begin by proving the optimality part
\begin{equation}
\label{sc rate optimal}\mathrm{\underline{sc}}(r)\geq\liminf_{n\to+\infty
}-\frac{1}{n}\log(1-\alpha_{n,r}^{\ad})
\ge
\sup_{\alpha>1}\inf_{\psi_{RA}}\frac{\alpha-1}{\alpha}\left[  
r-\widetilde{D}_{\alpha}(  {\mathcal{N}}_{A\rightarrow B}(\psi_{RA})\Vert\psi_{R}\otimes\sigma_{B})
\right]  .
\end{equation}

Note that if $S_{n},\,n\in\mathbb{N}$, is a sequence of adaptive strategies
such that $\liminf_{n\to\infty}-\frac{1}{n}\log\beta_{n}(S_{n})>r$ then for
all large enough $n$, $\beta_{n}(S_{n})\le2^{-nr}$, and thus $1-\alpha
_{n}(S_{n})\le1-\alpha_{n,r}^{\ad}$, which yields the first inequality
in \eqref{sc rate optimal}.

To prove the second inequality in \eqref{sc rate optimal}, 
consider the output states $\rho_{R_nB_n},\tau_{R_{n}B_{n}}$, and the test 
$Q_{R_{n}B_{n}}$, at the end of the adaptive discrimination strategy.
By Lemma \ref{lem:sandwich-to-htre} and \eqref{Renyi-CB2}, we get 
\begin{equation}
\label{Nagaoka bound}\frac{1}{n}\log\Tr\left\{  Q_{R_{n}B_{n}}\rho_{R_{n}%
B_{n}}\right\}  \leq\frac{\alpha-1}{\alpha}\left[  \frac{1}{n}\log\Tr\left\{
Q_{R_{n}B_{n}}\tau_{R_{n}B_{n}}\right\}  +\sup_{\psi_{RA}}\widetilde
{D}_{\alpha}(  {\mathcal{N}}(\psi_{RA})\Vert{\mathcal{R}}(\psi
_{RA}))  \right]  .
\end{equation}
Taking the supremum of both sides of \eqref{Nagaoka bound} over all strategies
such that the Type~II error is at most $2^{-nr}$, we obtain
\begin{equation}
\frac{1}{n}\log(1-\alpha_{n,r}^{\ad}) \leq\frac{\alpha-1}{\alpha
}\left[  -r +\sup_{\psi_{RA}}\widetilde{D}_{\alpha}(  {\mathcal{N}}%
(\psi_{RA})\Vert{\mathcal{R}}(\psi_{RA}))  \right] , \label{eq:ineq-to-gen}
\end{equation}
which yields the second inequality in \eqref{sc rate optimal}.


We now establish the achievability part
\begin{equation}
\label{scs proof2}\limsup_{n\to+\infty}-\frac{1}{n}(1-\alpha_{n,r}%
^{\ad})\le\mathrm{\overline{sc}}(r)\leq\sup_{\alpha>1}\inf_{\psi_{RA}%
}\frac{\alpha-1}{\alpha}\left[  r-\widetilde{D}_{\alpha}(  {\mathcal{N}%
}_{A\rightarrow B}(\psi_{RA})\Vert\psi_{R}\otimes\sigma_{B})
\right]  .
\end{equation}
The first inequality follows the same way as the first inequality in
\eqref{sc rate optimal}. Let $R$ be an arbitrary system. According to Theorem IV.10 and Remark IV.11 in
\cite{MO},
for every state $\psi_{RA}\in\S(\hil_R\otimes\hil_A)$ and every $r^{\prime}>0$,
there exists a sequence of tests $Q_{R_{n}B_{n}}$,$\,n\geq1$, such that
\begin{align}
\limsup_{n\rightarrow+\infty}\frac{1}{n}\log\Tr\left\{  Q_{R_{n}B_{n}%
}{\mathcal{R}}(\psi_{RA})^{\otimes n}\right\}   & \leq-r^{\prime
},\mbox{ }\mbox{ }\mbox{ }\mbox{ }\mbox{ }\mbox{ }\text{ and }\\
\liminf_{n\rightarrow+\infty}\frac{1}{n}\log\Tr\left\{  Q_{R_{n}B_{n}%
}{\mathcal{N}}(\psi_{RA})^{\otimes n}\right\}   &  \geq-H_{r^{\prime}}%
^{*}({\mathcal{N}}(\psi_{RA})\|{\mathcal{R}}(\psi_{RA})).
\end{align}
Thus,
\begin{equation}
\label{scs proof1}\mathrm{\overline{sc}}(r)\leq\inf_{r^{\prime}>r}%
H_{r^{\prime}}^{*}({\mathcal{N}}(\psi_{RA})\|{\mathcal{R}}(\psi_{RA})).
\end{equation}
From the definition \eqref{eq:sc-states} of the Hoeffding anti-divergence,
it is clear that $r\mapsto H_{r}^{*}({\mathcal{N}}(\psi_{RA})\|{\mathcal{R}%
}(\psi_{RA}))$ is a monotone increasing convex function on $(0,+\infty)$.
Moreover, Lemma IV.9 in \cite{MO} implies that $H_{r}^{*}({\mathcal{N}}%
(\psi_{RA})\|{\mathcal{R}}(\psi_{RA}))$ is finite for every $r>0$. Thus,
$r\mapsto H_{r}^{*}({\mathcal{N}}(\psi_{RA})\|{\mathcal{R}}(\psi_{RA}))$ is
continuous on $(0,+\infty)$, and \eqref{scs proof1} yields
\begin{align}
\mathrm{\overline{sc}}(r)\leq H_{r}^{*}({\mathcal{N}}(\psi_{RA})\|{\mathcal{R}%
}(\psi_{RA})).
\end{align}
Since this is true for every $\psi_{RA}$, we finally get
\begin{equation}
\mathrm{\overline{sc}}(r)\leq\inf_{\psi_{RA}}\sup_{\alpha>1}\frac{\alpha
-1}{\alpha}\left[  r-\widetilde{D}_{\alpha}(  {\mathcal{N}}(\psi
_{RA})\Vert{\mathcal{R}}(\psi_{RA}))  \right]  . \label{sc upper}%
\end{equation}

The last step is to show that the RHS of \eqref{sc rate optimal} and
\eqref{sc upper} are equal to each other. 
First, note that the RHS of \eqref{sc rate optimal} can be written as 
\begin{align*}
\sup_{\alpha>1}\inf_{\rho_{A'}}\frac{\alpha-1}{\alpha}
\left[  
r-\widetilde D_{\alpha}\!\bz\rho_{A'}^{1/2}\N_1(\Gamma_{A'A})\rho_{A'}^{1/2}\big\|\rho_{A'}\otimes\sigma_B\jz
\right],
\end{align*}
where the infimum is taken over $\S(\hil_{A'})$ with $A'\cong A$,
due to Lemma \ref{lemma:channel Renyi}. Moreover, the RHS of \eqref{sc upper} can be trivially upper bounded by 
\begin{align*}
\inf_{\rho_{A'}}\sup_{\alpha>1}\frac{\alpha-1}{\alpha}\left[  
r-\widetilde D_{\alpha}\!\bz\rho_{A'}^{1/2}\N_1(\Gamma_{A'A})\rho_{A'}^{1/2}\big\|\rho_{A'}\otimes\sigma_B\jz
\right]
\end{align*}
(see the proof of Lemma \ref{lemma:channel Renyi} in Appendix \ref{app:CB}).
Next, define
\begin{equation}
F(\alpha,\rho_{A'})\equiv(\alpha-1)
\widetilde D_{\alpha}\!\bz\rho_{A'}^{1/2}\N_1(\Gamma_{A'A})\rho_{A'}^{1/2}\big\|\rho_{A'}\otimes\sigma_B\jz
\end{equation}
for $\alpha>1$ and $\rho_{A'}\in\S(\hil_{A'})$.
Introducing the new variable $u\equiv\frac{\alpha-1}{\alpha}$, we have to show
that
\begin{equation}
\sup_{0<u<1}\inf_{\rho_{A'}}f(u,\rho_{A'})=\inf_{\rho_{A'}}\sup_{0<u<1}f(u,\rho_{A'}), \label{minimax}%
\end{equation}
where
\begin{equation}
f(u,\rho_{A'})\equiv ur-\widetilde{F}(u,\rho_{A'}),\mbox{ }\mbox{ }\mbox{ }\mbox{ }\mbox{ }\mbox{ }
\widetilde{F}(u,\rho_{A'})\equiv(1-u)F\left(  \frac{1}{1-u},\rho_{A'}\right)  .
\end{equation}
By Lemmas 3 and 4 in \cite{GW}, $\rho_{A'}\mapsto F(u,\rho_{A'})$ is concave, and
hence $\rho_{A'}\mapsto f(u,\rho_{A'})$ is convex, for any fixed $u\in(0,1)$.
On the other hand, $u\mapsto F(u,\rho_{A'})$ is convex by Corollary 3.11 in \cite{MO},
and Lemma \ref{lemma:convexity} below yields that $u\mapsto\widetilde{F}(u,\rho_{A'})$ is also convex, 
which in turn implies the concavity of $u\mapsto f(u,\rho_{A'})$ for any fixed $\rho_{A'}$.
By assumption, $D(\N\|\R)<+\infty$, and taking into account Lemma~\ref{lemma:channel Renyi finite},
it is easy to see that 
$\rho_{A'}\mapsto f(u,\rho_{A'})$ is continuous for any $u\in(0,1)$. Since the state space of $\hil_A$ is compact, the Kneser-Fan minimax theorem
\cite{Kneser,Fan} yields \eqref{minimax}.
\end{proof}
\smallskip


\begin{lemma}\label{lemma:convexity}
Let $f:\,(0,1)\to\bR$ be a convex function. Then 
\begin{align*}
\tilde f:\,u\mapsto (1-u)f\bz\frac{1}{1-u}\jz
\end{align*}
is convex as well.
\end{lemma}
\begin{proof}
Since $f$ is convex, it can be written as the supremum of affine functions, i.e., 
$f(x)=\sup_{i}\{a_ix+b_i\}$ for some $a_i,b_i\in\bR$, and thus
\begin{align*}
\tilde f(u)=(1-u)\sup_{i}\left\{a_i\frac{1}{1-u}+b_i\right\}
=
\sup_{i}\left\{a_i+b_i(1-u)\right\}.
\end{align*}
As a supremum of affine functions, $\tilde f$ is convex.
\end{proof}

\begin{remark}
It is not too difficult to see that Theorem \ref{stein} can be reformulated the following way:
\smallskip

\noindent [Direct part] For every $r<D(\N\|\R)$, there exists a sequence of adaptive strategies such that 
the type I error goes to $0$ and the type II error decays exponentially with a rate at least $r$.
\smallskip

\noindent [Strong converse part] For every $r>D(\N\|\R)$, and any sequence of adaptive strategies such that 
the type II error decays exponentially with a rate at least $r$, the type I error goes to $1$.
\smallskip

As we have seen, the direct part is an immediate consequence of Stein's lemma for state discrimination. 
For the proof of the strong converse part and for the proof of the optimality part of Theorem~\ref{scrtheorem}, we  
followed the same argument
of first using the monotonicity of the R\'enyi relative entropies under measurements and then applying
\eqref{Renyi-CB}. In fact, one could first prove the optimality part of Theorem \ref{scrtheorem} and obtain
the optimality part of Theorem \ref{stein} from it in the limit $r\searrow D(\N\|\R)$. Indeed, 
Lemma \ref{MI limits} implies that for any $r>D(\N\|\R)$, there exists an $\alpha>1$ such that 
$r>\wt D_{\alpha}(\N\|\R)$, and hence
the RHS of \eqref{sc rate optimal} is strictly positive, from which the strong converse part of the channel Stein's lemma is immediate.
\end{remark}

\subsection{Related results}

\label{sec:extension}

Hayashi and Tomamichel recently published their independently obtained results 
about a hypothesis testing scenario somewhat similar to
ours \cite{HT}. Both our paper and theirs generalise the task of binary 
state discrimination but in different and not directly comparable directions. They consider the problem of composite hypothesis testing, 
where the null hypothesis is the presence of a fixed bipartite state and the alternative hypothesis is the presence 
of a product state that shares one marginal with the null hypothesis. Considered as a channel discrimination problem, the null hypothesis is that the i.i.d.~channel
$\mathcal{N}_1^{\otimes n}$ is applied to the $A$ systems of the input, where 
the input state is restricted to be a fixed tensor-power state of the form $\psi_{RA}^{\otimes n}$. 
The alternative hypothesis is that a general ``worst-case'' replacer channel is applied to the $A$ systems, 
which leads to an output $\psi_R^{\otimes n} \otimes \sigma_{B^n}$, where $\sigma_{B^n}$ could be any state on the
$B$ systems. Not only do they allow for this more general alternative
hypothesis, but they also determine both the direct and the strong converse exponents in their scenario. On the other hand, one has to note that when the above result is considered as a channel discrimination problem, allowing only the tensor powers of one fixed state as an input is extremely restrictive. In contrast, our results do allow for more general input states and for the adaptive strategies that distinguish the problem 
of quantum channel discrimination from binary state discrimination.

While the results of the two papers go in quite different directions,
 there is also a natural combination of them, which enables us to obtain a Stein's lemma with strong converse for the following channel discrimination problem with composite alternative hypothesis.
For every $n\in\bN$, the null hypothesis is that the channel is $\N^{\otimes n}$, where $\N$ is a fixed channel, and the alternative hypothesis is that the channel belongs to the set $\R^{(n)}\equiv\{\R_{\sigma_n}:\,\sigma_n\in\Sigma_n\}$, where 
\begin{align}
\{\sigma^{\otimes n}:\,\sigma\in\S(\hil_B)\}\subseteq
\Sigma_n\subseteq
\S(\hil_B^{\otimes n}).
\end{align}
For Stein's lemma, one is interested in the asymptotics of the optimal Type II error
\begin{align}
\beta_{\ep}\x\equiv\beta_{\ep}\x(\N^{\otimes n}\|\R^{(n)})\equiv\inf\left\{\sup_{\sigma_n\in\Sigma_n}\beta_n(S_n|\sigma_n):\,\alpha_n(S_n)\le \ep\right\},
\end{align}
where the infimum is over all strategies in the class $x$ with Type I error below $\ep$. 
Combining Theorem 11 in \cite{HT} and Theorem \ref{stein} in this paper, we obtain the following:

\begin{theorem}\label{thm:Stein comp}
In the above setting, for every $\ep\in(0,1)$,
\begin{align}
\lim_{n\to+\infty}-\frac{1}{n}\log\beta_{\ep}^{\ad}(\N^{\otimes n}\|\R^{(n)})
=
\lim_{n\to+\infty}-\frac{1}{n}\log\beta_{\ep}^{\prod}(\N^{\otimes n}\|\R^{(n)})
&= \inf_{\sigma_B\in\S(\hil_B)} D(\N\|\R_{\sigma_B})\label{Stein exponent1}\\
&=I(\N),\label{Stein exponent2}
\end{align}
where $ D(\N\|\R_{\sigma_B})$ is the channel relative entropy \eqref{channel Stein}, and 
$I(\N)$ is the channel mutual information \eqref{IN}.
\end{theorem}
\begin{proof}
Just as in \eqref{composite3}--\eqref{composite4} (see below), we have
\begin{align}
\beta_{\ep}^{\ad}
\le
\beta_{\ep}^{\prod}
&=
\inf_{\psi_{RA}}\inf_{Q_n}\left\{\sup_{\sigma_n\in\Sigma_n}\Tr \left\{ Q_n(\psi_R^{\otimes n}\otimes\sigma_n)\right\}:\,
\Tr \left\{ (I_n-Q_n)\bz\N_{A\to B}(\psi_{RA})\jz^{\otimes n} \right\}\le\ep\right\}\\
&\le
\inf_{Q_n}\left\{\sup_{\sigma_n\in\Sigma_n}\Tr \left\{ Q_n(\psi_R^{\otimes n}\otimes\sigma_n)\right\}:\,
\Tr \left\{ (I_n-Q_n)\bz\N_{A\to B}(\psi_{RA})\jz^{\otimes n} \right\}\le\ep\right\}\\
&\equiv\beta_{\ep}(\psi_{RA}),\label{composite5}
\end{align}
where $Q_n$ runs over all $Q_n\in\B(\hil_{RA}^{\otimes n})_+$ such that $Q_n\le I_n$, and the second inequality 
holds for every $\psi_{RA}$. By \cite[Theorem 11]{HT}, for any $\psi_{RA}$ and any rate 
$r$, there exists a sequence of binary measurements $(Q_n,I_n-Q_n)$,
for which
\begin{align}
\sup_{\sigma_n\in\S(\hil_B^{\otimes n})}\Tr \left\{ Q_n(\psi_R^{\otimes n}\otimes\sigma_n)\right\}
&\le 
2^{-nr},\\
\limsup_{n\to+\infty}\frac{1}{n}\log\Tr \left\{ (I_n-Q_n)\bz\N_{A\to B}(\psi_{RA})\jz^{\otimes n} \right\}
&=
-\sup_{\alpha\in(0,1)}\frac{\alpha-1}{\alpha}\left[r-I_{\alpha}(R;B)_{\N(\psi)}\right].\label{Hbound}
\end{align}
By Lemma \ref{MI limits}, the RHS of \eqref{Hbound} is strictly negative for every 
$r<I(R;B)_{\N(\psi)}$, and hence
\begin{align}
\limsup_{n\to+\infty}\frac{1}{n}\log\beta_{\ep}(\psi_{RA})\le -I(R;B)_{\N(\psi)}.
\end{align}
When combined with \eqref{composite5}, this yields
\begin{align}
\limsup_{n\to+\infty}\frac{1}{n}\log\beta_{\ep}^{\ad}\le
\limsup_{n\to+\infty}\frac{1}{n}\log\beta_{\ep}^{\prod}\le
 \inf_{\psi_{RA}}-I(R;B)_{\N(\psi)}=-I(\N).
\end{align}

Suppose now that 
\begin{align}
\liminf_{n\to+\infty}\frac{1}{n}\log\beta_{\ep}^{\ad}<-r
\end{align}
for some $r\in\bR$. For every $\sigma\in\S(\hil_B)$, 
$\beta_{\ep}^{\ad}=\beta_{\ep}^{\ad}(\N^{\otimes n}\|\R^{(n)})\ge\beta_{\ep}^{\ad}(\N^{\otimes n}\|\R_{\sigma}^{\otimes n})$.
Hence the assumption yields that $\liminf_{n\to+\infty}\frac{1}{n}\log\beta_{\ep}^{\ad}(\N^{\otimes n}\|\R_{\sigma}^{\otimes n})<-r$, and by Theorem \ref{stein} this is only possible if $r\le D(\N\|\R_{\sigma})$. Since this is true 
for every $\sigma\in\S(\hil_B)$, we finally get that 
$r\le \inf_{\sigma\in\S(\hil_B)}D(\N\|\R_{\sigma})=I(\N)$, completing the proof of \eqref{Stein exponent1}. 

The equality
of \eqref{Stein exponent2} and \eqref{Stein exponent1} is due to Lemma \ref{lemma:geom int}.
\end{proof}
\bigskip

It is now natural to ask whether the exact strong converse exponent can be determined for this problem, analogously to Theorem \ref{scrtheorem}. Below we give lower and upper bounds for the strong converse exponent. We conjecture that these bounds in fact coincide, and thus give the exact strong converse exponent; indeed, this could be proved 
if one could justify interchanging the order of infima and suprema in \eqref{composite1} and \eqref{composite2} below.

The problem can be formulated as follows. For any adaptive discrimination strategy $S_n$, and any $\sigma_n\in\S(\hil_B^{\otimes n})$, let $\alpha_n(S_n)$ and 
$\beta_n(S_n|\sigma_n)$ be the Type I and Type II error 
probabilities for discriminating between $\N^{\otimes n}$ and $\R_{\sigma_n}$, as given in \eqref{error probs}.
We consider the optimal Type I error
\begin{align}\label{composite type I}
\alpha_{n,r}\x\equiv\alpha_{2^{-nr}}\x(\N^{\otimes n}\|\R^{(n)}) 
& \equiv
\inf\left\{ \alpha_{n}(S_{n}):\,\sup_{\sigma_n\in\Sigma_n}\beta_{n}(S_{n}|\sigma_n)\le2^{-nr}\right\},
\end{align}
where $x$ denotes the set of allowed discrimination strategies and the
optimisation is over all strategies in the class $x$. As before, we take $x=\prod$ and $x=\ad$, for product and adaptive strategies, respectively.
We have
\begin{align}
\alpha_{n,r}^{\prod}
&=
\inf_{\psi_{RA}}\inf_{Q_n}\left\{\Tr \left\{ (I_n-Q_n)\bz\N_{A\to B}(\psi_{RA})\jz^{\otimes n} \right\}:\,
\sup_{\sigma_n\in\Sigma_n}\Tr \left\{ Q_n(\psi_R^{\otimes n}\otimes\sigma_n) \right\} \le 2^{-nr}\right\}\label{composite3}\\
&\le
\inf_{\psi_{RA}}\inf_{Q_n}\left\{\Tr \left\{ (I_n-Q_n)\bz\N_{A\to B}(\psi_{RA})\jz^{\otimes n}\right\}:\,
\sup_{\sigma_n\in\S(\hil^{\otimes n})}\Tr \left\{ Q_n(\psi_R^{\otimes n}\otimes\sigma_n) \right\} \le 2^{-nr}\right\}\\
&\le
\inf_{Q_n}\left\{\Tr \left\{ (I_n-Q_n)\bz\N_{A\to B}(\psi_{RA})\jz^{\otimes n} \right\}:\,
\sup_{\sigma_n\in\S(\hil^{\otimes n})}\Tr \left\{ Q_n(\psi_R^{\otimes n}\otimes\sigma_n) \right\} \le 2^{-nr}\right\}\\
&\equiv \alpha_{n,r}(\psi_{RA}),\label{composite4}
\end{align}
where $Q_n$ runs over all $Q_n\in\B(\hil_{RA}^{\otimes n})_+$ such that $Q_n\le I_n$, and the second inequality 
holds for every $\psi_{RA}$.
Applying now the results of \cite{HT}, we get that 
\begin{align}
\limsup_{n\to+\infty}-\frac{1}{n}\log(1-\alpha_{n,r}^{\ad})
&\le
\limsup_{n\to+\infty}-\frac{1}{n}\log(1-\alpha_{n,r}^{\prod})\\
&\le
\inf_{\psi_{RA}}\limsup_{n\to+\infty}-\frac{1}{n}\log(1-\alpha_{n,r}(\psi_{RA}))\\
&=
\inf_{\psi_{RA}}\sup_{\alpha>1}\frac{\alpha-1}{\alpha}\left[r-\inf_{\sigma\in\S(\hil_B)}\wt D_{\alpha}\!\bz\N_{A\to B}(\psi_{RA})\|\psi_R\otimes\sigma \jz\right]\label{composite0}\\
&=
\inf_{\psi_{RA}}\sup_{\alpha>1}\sup_{\sigma\in\S(\hil_B)}\frac{\alpha-1}{\alpha}\left[r-\wt D_{\alpha}\!\bz\N_{A\to B}(\psi_{RA})\|\psi_R\otimes\sigma \jz\right],\label{composite1}
\end{align}
where \eqref{composite0} is due to \cite[Theorem 13]{HT}.
On the other hand,
\begin{align}
\alpha_{n,r}^{\ad}
&\ge
\inf\left\{ \alpha_{n}(S_{n}):\,\beta_{n}(S_{n}|\sigma^{\otimes n})\le2^{-nr}\right\}=
\alpha_{2^{-nr}}^{\ad}(\N^{\otimes n}\|\R_{\sigma}^{\otimes n}),
\ds\ds\ds\sigma\in\S(\hil_B),
\end{align}
and hence
\begin{align}
\liminf_{n\to+\infty}-\frac{1}{n}\log(1-\alpha_{n,r}^{\ad})
&\ge
\sup_{\sigma\in\S(\hil_B)}\liminf_{n\to+\infty}-\frac{1}{n}\log(1-\alpha_{2^{-nr}}^{\ad}(\N^{\otimes n}\|\R_{\sigma}^{\otimes n})\label{composite7}\\
&=
\sup_{\sigma\in\S(\hil_B)}
\sup_{\alpha>1}\inf_{\psi_{RA}}\frac{\alpha-1}{\alpha}\left[r-\wt D_{\alpha}\!\bz\N_{A\to B}(\psi_{RA})\|\psi_R\otimes\sigma \jz\right]\label{composite6}\\
&=
\sup_{\sigma\in\S(\hil_B)}
\inf_{\psi_{RA}}\sup_{\alpha>1}\frac{\alpha-1}{\alpha}\left[r-\wt D_{\alpha}\!\bz\N_{A\to B}(\psi_{RA})\|\psi_R\otimes\sigma \jz\right]\label{composite2}
\end{align}
where the two equalities are due to Theorem \ref{scrtheorem}. If one had joint concavity in the variables $\alpha >1$ and $\sigma$, then one could interchange the optima and show that \eqref{composite1} and \eqref{composite2} are equal to each other, obtaining strong converse exponents for this channel discrimination problem. However, it remains unclear to us if the joint concavity holds or more generally if the exchange is possible.

\begin{remark}
Theorem \ref{thm:Stein comp} gives an operational interpretation to the channel mutual information $I(\N)$, and its geometric 
representation given in Lemma \ref{lemma:geom int}. If \eqref{composite1} and \eqref{composite2} could be shown to be equal, that would give an analogous operational interpretation to the channel R\'enyi mutual informations $\wt I_{\alpha}(\N)$ and their geometric representation in Lemma \ref{lemma:geom int}, for every $\alpha>1$.
\end{remark}

\section{Strong converse for quantum-feedback-assisted classical communication}
\label{sec:EA sc proof}

In this section, we give a detailed proof of Theorem~\ref{thm:sc-feedback}, which
identifies a strong converse exponent for quantum-feedback-assisted
communication and states that a strong converse theorem holds for the
quantum-feedback-assisted classical capacity of a quantum channel.

In an $n$-round feedback-assisted protocol $\P_n$, Alice and Bob initially share an entangled state on Alice's system $X_0$ and Bob's system $B_0'$. If Alice wants to transmit message $m\in\{1,\ldots,M_n\}$, where $M_n\in\bN$ is the number of messages, she applies a quantum channel $\E^1_m$ with output system $A_1'A_1$ to her part of the entangled state; the result is a state $\rho^m_{A_1'A_1B_0'}=\tau^m_{A_1'A_1B_0'}$ on systems $A_1'A_1B_0'$, where $A_1$ is sent over the channel to Bob, with an output in system $B_1$, while $A_1'$ is kept at Alice's side for possible later use. After this, Bob may apply a quantum channel $\D^1$ on $B_1B_0'$ with an output on $X_1B_1'$, of which system $X_1$ contains the feedback information, that is sent back to Alice, while $B_1'$ is kept at Bob's side for possible later use. 
This procedure is repeated $n$ times, as depicted in Figure~\ref{FeedbackTest} (with $n=3$).
At each round, an encoding channel $\E^i_m:\,A_{i-1}'X_{i-1}\to A_i'A_i$ corresponding to the same fixed message $m$ is applied, but the $\E^i_m$ may be different channels for different $i$'s.
At the end of the protocol, the $\D^n$ channel is a POVM  on $B_{n}B_{n-1}'$ with outcomes in $\{1,\ldots,M_n\}$, specified by the POVM elements $\{D^m_{B_nB_{n-1}'}\}_{m=1}^{M_n}$. In the last round, $A_{n}'$ can be taken one-dimensional, since whatever information may be stored there does not influence the outcome of the final measurement on Bob's systems.
We assume for simplicity that the feedback channel is noiseless, although it is not necessary to do so; indeed, we are looking for an upper bound on the success probability, and noisy feedback can only decrease the success probability.

For every stage of the communication process, let $\rho^m$ with the appropriate labels denote the state obtained from 
$\rho^m_{A_1'A_1B_0'}$ by the action of all channels $\E^i_m,\N^i,\D^i$ up to that stage; e.g., 
$\rho^m_{A_1'B_1B_0'}=
\N_{A_1\to B_1}(\rho^m_{A_1'A_1B_0'})$, etc.
Similarly, let $\tau^m$ with the appropriate labels denote the state obtained from $\tau^m_{A_1'A_1B_0'}$ up to a certain stage of the process, where all uses 
of $\N$ are replaced by $\R_{\sigma}$ for some fixed state $\sigma$; see Figure~\ref{FeedbackTestReplacer}.
Moreover, we introduce an auxiliary 
system $R$ with an orthonormal basis $\{\ket{m}_{R}\}_{m=1}^{M_n}$, and define
\begin{align}
\rho_{RL}\deff\frac{1}{M_n}\sum_{m=1}^{M_n}\pr{m}_R\otimes \rho^m_{L},\ds\ds\ds\ds\ds\ds
\tau_{RL}\deff\frac{1}{M_n}\sum_{m=1}^{M_n}\pr{m}_R\otimes \tau^m_{L},\
\end{align} 
where $L$ is any set of indices that can occur at a certain stage of the communication process, and
\begin{align}
T_{RB_nB_{n-1}'}\deff\sum_{m=1}^{M_n}\pr{m}_R\otimes D^m_{B_nB_{n-1}'},
\end{align} 
such that $0\le T_{RB_nB_{n-1}'}\le I_{RB_nB_{n-1}'}$. For every $1\le i\le n$, we define
$\E^i:\,RA_{i-1}'X_{i-1}\to RA_{i}'A_{i}$ as
\begin{align}
\E^i\bz\sum_{m=1}^{M_n}\pr{m}_R\otimes \psi_{A_{i-1}'X_{i-1}}\jz\deff
\sum_{m=1}^{M_n}\pr{m}_R\otimes \E^i_m( \psi_{A_{i-1}'X_{i-1}}).
\end{align}

\begin{figure}[ptb]
\begin{center}
\includegraphics[width=6in]{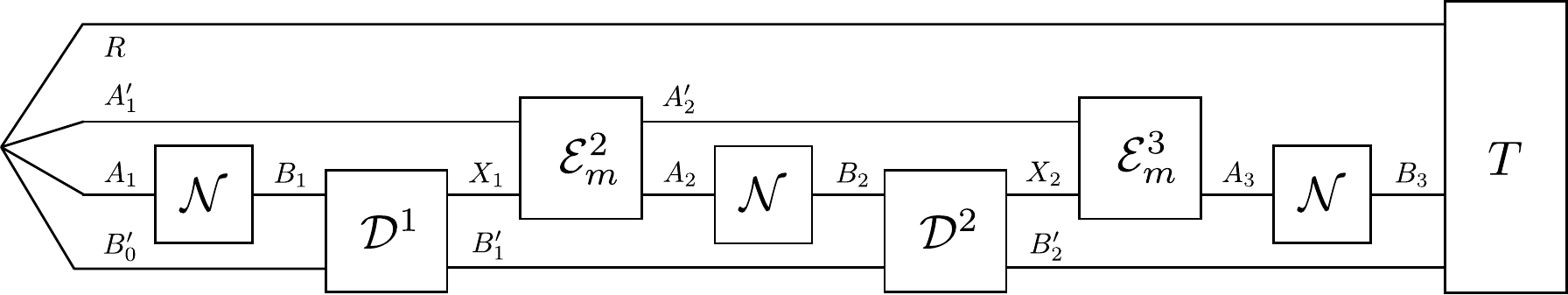}
\end{center}
\caption{A general quantum feedback-assisted communication protocol for the channel $\N$.}
\label{FeedbackTest}%
\end{figure}

\begin{figure}[ptb]
\begin{center}
\includegraphics[width=6in]{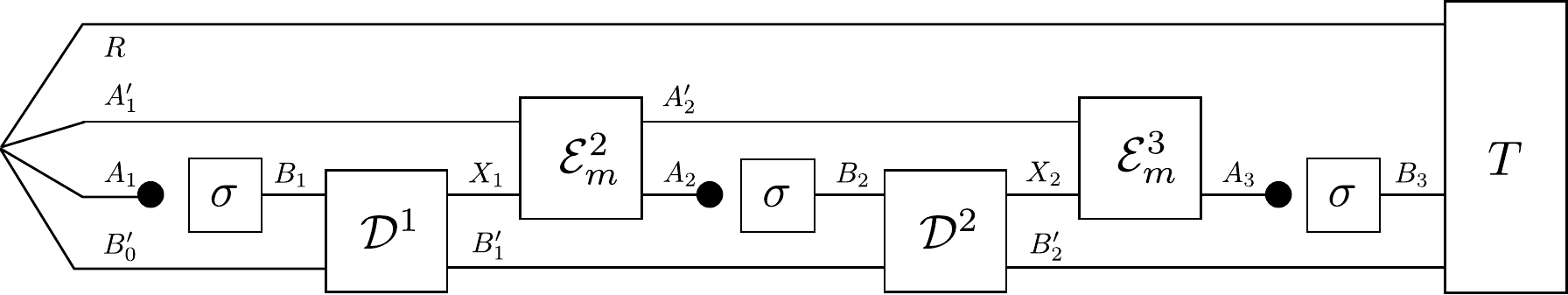}
\end{center}
\caption{A general quantum feedback-assisted communication protocol, with a replacer channel instead of the channel $\N$.}
\label{FeedbackTestReplacer}%
\end{figure}

If the outcome of the final measurement $\D^n$ is $m'$ then Bob concludes that the message $m'$ was sent. 
The success probability $\ps(\P_n)$ of the protocol is given by 
\begin{align}
\ps(\P_n)\equiv\frac{1}{M_n}\sum_{m=1}^{M_n}\Tr\left\{D^m_{B_nB_{n-1}'} \rho^m_{B_nB_{n-1}'} \right\}
=
\Tr\left\{T_{RB_nB_{n-1}'}\rho_{RB_nB_{n-1}'}\right\}. 
\end{align}

Note that for every round $k$, $\tau^m_{B_kB_{k-1}'}$ is independent of $m$, and we have
\begin{align}
\tau_{RB_kB_{k-1}'}=\tau_R\otimes\sigma_{B_k}\otimes\tau_{B_{k-1}'},\ds\ds\ds\text{where}\ds\ds\ds
\tau_R=\frac{1}{M_n}I_{R}.
\end{align}
This is because all information about 
the identity of the message is kept at Alice's side all through the protocol, as one can easily see in Figure \ref{FeedbackTestReplacer}. Hence,
\begin{align}
\Tr\left\{T_{RB_nB_{n-1}'}\tau_{RB_nB_{n-1}'}\right\}
=\frac{1}{M_n}\sum_{m=1}^{M_n}\Tr\left\{D^m_{B_nB_{n-1}'}(\sigma_{B_n}\otimes \tau_{B_{n-1}'})\right\}
=\frac{1}{M_n}.
\end{align}
Now we can apply Nagaoka's method \cite{N}, and use the
 monotonicity of $\wt D_{\alpha}$ to get
\begin{align}
\wt D_{\alpha}(\rho_{RB_nB_{n-1}'}\|\tau_{RB_nB_{n-1}'}) &
\ge \frac{1}{\alpha-1}\log \left[ \ps(\P_n)^{\alpha}\bz\frac{1}{M_n}\jz^{1-\alpha} \right]\\
& =\frac{\alpha}{\alpha-1}\log \ps(\P_n)+\log M_n.
\end{align}

We will use the same iterative method as in Section \ref{sec:Stein proof} to complete the proof of Theorem \ref{thm:sc-feedback}.
For every $k>1$,
\begin{align}
&\wt D_{\alpha}(\rho_{RA_k'B_kB_{k-1}'}\|\tau_{RA_k'B_kB_{k-1}'})\\
&\ds=
\wt D_{\alpha}(\N_{A_k\to B_k}(\rho_{RA_k'A_kB_{k-1}'})\|\tau_{RA_k'B_{k-1}'}\otimes\sigma_{B_k})\\
&\ds=
\frac{\alpha}{\alpha-1}\log\norm{\bz\Theta_{\sigma_{B_k}^{\frac{1-\alpha}{\alpha}}}\circ \N_{A_k\to B_k}\jz
\bz \tau_{RA_k'B_{k-1}'}^{\frac{1-\alpha}{2\alpha}}\rho_{RA_k'A_kB_{k-1}'}\tau_{RA_k'B_{k-1}'}^{\frac{1-\alpha}{2\alpha}}\jz}_{\alpha}\\
&\ds\le
\frac{\alpha}{\alpha-1}\log\sup_{X_{RA_k'A_kB_{k-1}'}\ge 0}
\frac{
\norm{\bz\Theta_{\sigma_{B_k}^{\frac{1-\alpha}{\alpha}}}\circ \N_{A_k\to B_k}\jz
X_{RA_k'A_kB_{k-1}'}}_{\alpha}
}
{\norm{X_{RA_k'B_{k-1}'}}_{\alpha}}
\norm{ \tau_{RA_k'B_{k-1}'}^{\frac{1-\alpha}{2\alpha}}\rho_{RA_k'B_{k-1}'}\tau_{RA_k'B_{k-1}'}^{\frac{1-\alpha}{2\alpha}}}_{\alpha}\\
&=
\frac{\alpha}{\alpha-1}\log\cbnorm{\Theta_{\sigma_{B_k}^{\frac{1-\alpha}{\alpha}}}\circ \N_{A_k\to B_k}}
+
\wt D_{\alpha}(\rho_{RA_k'B_{k-1}'}\|\tau_{RA_k'B_{k-1}'}).
\end{align}
Now, if $k=1$ then $\rho_{RA_k'B_{k-1}'}=\tau_{RA_k'B_{k-1}'}$ by definition, and the last term above is zero. Otherwise we can upper bound the last term above as
\begin{align}
\wt D_{\alpha}(\rho_{RA_k'B_{k-1}'}\|\tau_{RA_k'B_{k-1}'})
\le
\wt D_{\alpha}(\rho_{RA_{k-1}'B_{k-1}B_{k-2}'}\|\tau_{RA_{k-1}'B_{k-1}B_{k-2}'}),
\end{align}
where the inequality is due to the monotonicity of $\wt D_{\alpha}$ under $\Tr_{A_k} \circ \E^k \circ \D^{k-1}$.

Using the above steps iteratively, we finally get
\begin{align}
\frac{\alpha}{\alpha-1}\log \ps(\P_n)+\log M_n
&\le
\wt D_{\alpha}(\rho_{RB_nB_{n-1}'}\|\tau_{RB_nB_{n-1}'}) \\
& \le 
n\frac{\alpha}{\alpha-1}\log\cbnorm{\Theta_{\sigma^{\frac{1-\alpha}{2\alpha}}}\circ \N}\\
&=
n\sup_{\psi_{\hat AA}}\widetilde{D}_{\alpha}(  \mathcal{N}_{A\rightarrow
B}(\psi_{\hat AA})\Vert\psi_{\hat A}\otimes\sigma_{B}),  
\end{align}
where the last identity is due to \eqref{lemmanorm}, and the supremum is over all pure states on $\hat AA$, where
$\hat A$ is a copy of $A$. Since this is true for every $\sigma_B\in\S(\hil_B)$, we get
\begin{align}\label{EA sc bound0}
\frac{\alpha}{\alpha-1}\frac{1}{n}\log \ps(\P_n)+\frac{1}{n}\log M_n
&\le
\inf_{\sigma_B\in\S(\hil_B)}\sup_{\psi_{\hat AA}}\widetilde{D}_{\alpha}\left(  \mathcal{N}_{A\rightarrow
B}(\psi_{\hat AA})\Vert\psi_{\hat A}\otimes\sigma_{B}\right)=\wt I_{\alpha}(\N).
\end{align}
Hence, if $\P_n$ is a feedback-assisted coding scheme such that  
$\frac{1}{n}\log M_n\ge R$ then 
\begin{align}\label{EA sc bound}
\frac{1}{n}\log\ps(\P_n)\le -\sup_{\alpha>1}\frac{\alpha-1}{\alpha}\bz R-\wt I_{\alpha}(\N)\jz,
\end{align}
where we used that \eqref{EA sc bound0} holds for every $\alpha>1$. This proves \eqref{eq:QFA-strong-conv-exponent}
of Theorem \ref{thm:sc-feedback}.

By \eqref{EA sc bound}, the success probability goes to zero exponentially fast for any rate $R>R_{\min}\equiv\inf_{\alpha>1}\wt I_{\alpha}(\N)$. 
By the monotonicity of the R\'enyi relative entropies in $\alpha$, 
$\inf_{\alpha>1}\wt I_{\alpha}(\N)=\lim_{\alpha\searrow 1}\wt I_{\alpha}(\N)$, 
and the latter is equal to $I_1(\N)=I(\N)$, due to Lemma~\ref{MI limits}.
This proves the last assertion of Theorem~\ref{thm:sc-feedback}.

\section{Conclusion}

This paper establishes a quantum Stein's lemma and identifies the strong
converse exponent when discriminating an arbitrary channel from the replacer
channel. The conclusion is that a tensor-power, non-adaptive strategy is
optimal in this regime. This result has implications in the physical setting
of quantum illumination, as described in Section~\ref{sec:quantum-illum}. We
have also proven that a strong converse theorem holds in the setting of
quantum-feedback-assisted communication, strengthening a weak converse result
due to Bowen \cite{B04}. This strong converse theorem also strengthens the
main result of \cite{GW}, in which a bound on the strong converse exponent was
established for the entanglement-assisted communication setting. We show here
that this same bound holds in the more general quantum-feedback-assisted
communication setting. We also briefly discussed how to combine
our results in adaptive channel discrimination with those of Hayashi and Tomamichel
from \cite{HT} to obtain a quantum Stein's lemma in a more general setting than that
considered in either paper. It remains an open question to determine the strong converse
exponent for this more general setting.

There are several other open questions to consider going forward from here. First,
is the strong converse exponent bound in (\ref{eq:QFA-strong-conv-exponent})
optimal? That is, does there exist an entanglement-assisted communication
protocol that achieves this bound? Encouraging for us here is that a full solution
is known for the classical version of this problem \cite{A78,1056003,CK82}.
Next, can we say anything about the direct
domain for either the adaptive channel discrimination setting or the
quantum-feedback-assisted communication setting? Any results obtained in the
latter setting would be a counterpart to the error exponents found in
\cite{BH98,H00,PhysRevA.76.062301} for classical communication. Finally, would
the conclusions of this paper extend to the setting of symmetric hypothesis
testing? That is, would it be possible to show that non-adaptive strategies
suffice here?
\bigskip

\textit{Acknowledgements}---We thank Nilanjanna Datta, Manish K.~Gupta,
Bhaskar Roy Bardhan, and
Marco Tomamichel for insightful discussions on these topics. We thank David Ding for feedback on the manuscript. TC\ and
MMW\ acknowledge support from the Department of Physics and Astronomy at LSU,
from the NSF\ under Award No.~CCF-1350397, and from the DARPA Quiness Program
through US Army Research Office Award No.~W31P4Q-12-1-0019.
MM acknowledges support from the European Research Council
Advanced Grant \textquotedblleft IRQUAT\textquotedblright,
the Spanish MINECO (Project No. FIS2013-40627-P), the Generalitat de Catalunya CIRIT (Project No. 2014 SGR 966),
the Hungarian Research Grant OTKA-NKFI K104206,
and the Technische Universit\"at M\"unchen -- Institute for Advanced Study, funded by the German Excellence Initiative and the European Union Seventh Framework Programme under grant agreement no. 291763.

\appendix

\section{Channel divergences}
\label{app:CB}

\noindent\textbf{Proof of Lemma \ref{lemma:channel Renyi0}.}
We only prove the assertion for $\wt D_{\alpha}$, as the proof for $D_{\alpha}$ goes the same way.
%

For every system $R$, $\ket{\Gamma_{A'A}}$ defines an isomorphism between 
$\hil_{R}\otimes\hil_A$ and $\B(\hil_{A'},\hil_R)$, under which 
$X\in\B(\hil_{A'},\hil_R)$ corresponds to $(X\otimes I_A)\ket{\Gamma_{A'A}}=\sum_i (Xe_i)\otimes e_i$. 
Given a pure state $\psi_{RA}$, it can be written as  
$\psi_{RA}=\pr{\psi_{RA}}$, with $\ket{\psi_{RA}}=(X\otimes I_A)\ket{\Gamma_{A'A}}$, where 
$X\in\B(\hil_{A'},\hil_R)$, and $\Tr X^*X=1$.
Thus, for any channel $\mathcal{N}_{A\to B}$, 
\begin{equation}
\N_{A\to B}(\psi_{RA})=\N\left((X\otimes I_A)\Gamma_{A'A}(X^*\otimes I_A)\right)
=(X\otimes I_A)\N\left(\Gamma_{A'A}\right)(X^*\otimes I_A).
\end{equation}
Let $V:\,\hil_{A'}\to\hil_R$ be a partial isometry such that $X=V|X|$. Then it is easy to see that
\begin{align}
& \wt  D_{\alpha}\bz \N_1(\psi_{RA})\|\N_2(\psi_{RA})\jz \\
&=
\wt D_{\alpha}\bz (X\otimes I_A)\bz\N_1(\Gamma_{A'A})\jz(X^*\otimes I_A)\|(X\otimes I_A)\bz\N_2(\Gamma_{A'A})\jz(X^*\otimes I_A)
\jz\\
&=
\wt D_{\alpha}\bz (V|X|\otimes I_A)\bz\N_1(\Gamma_{A'A})\jz(|X|V^*\otimes I_A)\|(V|X|\otimes I_A)\bz\N_2(\Gamma_{A'A})
\jz(|X|V^*\otimes I_A)\jz\\
&=
\wt D_{\alpha}\bz (|X|\otimes I_A)\bz\N_1(\Gamma_{A'A})\jz(|X|\otimes I_A)\|(|X|\otimes I_A)\bz\N_2(\Gamma_{A'A})\jz(|X|\otimes I_A)\jz\\
&=
\wt D_{\alpha}\bz \rho_{A'}^{1/2}\N_1(\Gamma_{A'A})\rho_{A'}^{1/2}\|\rho_{A'}^{1/2}\N_2(\Gamma_{A'A})\rho_{A'}^{1/2}\jz\\
&=
\widetilde D_{\alpha}\bz\N_1\bz\rho_{A'}^{1/2}\pr{\Gamma_{A'A}}\rho_{A'}^{1/2}\jz
\Big\|\N_2\bz\rho_{A'}^{1/2}\pr{\Gamma_{A'A}}\rho_{A'}^{1/2}\jz\jz,
\end{align}
where $\rho_{A'}:=|X|^2=X^*X$. 
The equality of the last two expressions follow from the fact that the channels only act on the $A$ system.
This completes the proof.
\bigskip

\noindent\textbf{Proof of Lemma \ref{lemma:CB}.} 
Let $\Theta$ denote the conjugation by $\sigma_B^{\frac{1-\alpha}{2\alpha}}$. 
With the notations in the proof of Lemma \ref{lemma:channel Renyi0}, every pure state 
$\psi_{A'A}$ can be written as $\psi_{A'A}=\pr{\psi_{A'A}}$,
$\ket{\psi_{A'A}}=(X\otimes I)\ket{\Gamma_{A'A}}$.
Then $\psi_{A'}=XX^*$, and hence
\begin{align}
&\wt D_{\alpha}\bz\N(\psi_{A'A})\|\R_{\sigma_B}(\psi_{A'A})\jz
=\wt D_{\alpha}\bz\N_{A\to B}(\psi_{A'A})\|\psi_{A'}\otimes \sigma_B\jz\\
&\ds=
\frac{1}{\alpha-1}\log\Tr
\left\{\left[(\psi_{A'}\otimes\sigma_B)^{\frac{1-\alpha}{2\alpha}}\bz\N_{A\to B}(\psi_{A'A})\jz(\psi_{A'}\otimes\sigma_B)^{\frac{1-\alpha}{2\alpha}}\right]^{\alpha}\right\}\\
&\ds=
\frac{1}{\alpha-1}\log\Tr\left\{\left[\bz (XX^*)^{\frac{1-\alpha}{2\alpha}}
\otimes\sigma_B^{\frac{1-\alpha}{2\alpha}}\jz\bz X\otimes I\jz \bz\N(\Gamma_{A'A})\jz\bz X^*\otimes I\jz
\bz(XX^*)^{\frac{1-\alpha}{2\alpha}}\otimes\sigma_B^{\frac{1-\alpha}{2\alpha}}\jz\right]^{\alpha}\right\}.
\end{align}
Let $X=V|X|$ for some unitary $V$; then $XX^*=V|X|^2V^*$, and
$(XX^*)^{\frac{1-\alpha}{2\alpha}}X=V|X|^{\frac{1-\alpha}{\alpha}}V^*V|X|=V|X|^{\frac{1}{\alpha}}$. Thus
\begin{align}
&\wt D_{\alpha}\bz\N_{A\to B}(\psi_{A'A})\|\psi_{A'}\otimes \sigma_B\jz\label{CB1} \nonumber\\
&\ds=
\frac{1}{\alpha-1}\log\Tr\left\{\left[
\bz V\otimes I_B\jz\bz |X|^{\frac{1}{\alpha}}\otimes\sigma_B^{\frac{1-\alpha}{2\alpha}}\jz\bz \N(\Gamma_{A'A})\jz
\bz |X|^{\frac{1}{\alpha}}\otimes\sigma_B^{\frac{1-\alpha}{2\alpha}}\jz\bz V^*\otimes I_B\jz\right]^{\alpha}\right\}\\
&\ds=
\frac{1}{\alpha-1}\log\Tr\left\{\left[
\bz |X|^{\frac{1}{\alpha}}\otimes\sigma_B^{\frac{1-\alpha}{2\alpha}}\jz\bz \N(\Gamma_{A'A})\jz
\bz |X|^{\frac{1}{\alpha}}\otimes\sigma_B^{\frac{1-\alpha}{2\alpha}}\jz\right]^{\alpha}\right\}\\
&\ds=
\frac{1}{\alpha-1}\log\Tr\left\{\left[
\bz |X|^{\frac{1}{\alpha}}\otimes I_B\jz(\Theta\circ\N)(\Gamma_{A'A})
\bz |X|^{\frac{1}{\alpha}}\otimes I_B\jz\right]^{\alpha}\right\}\\
&\ds=\frac{\alpha}{\alpha-1}\log\norm{
\bz Y^{\frac{1}{2\alpha}}\otimes I_B\jz(\Theta\circ\N)(\Gamma_{A'A})
\bz Y^{\frac{1}{2\alpha}}\otimes I_B\jz}_{\alpha},
\label{CB2}
\end{align}
where we used the notation $Y\equiv X^*X\in\B(\hil_{A'})$.
Hence, optimizing \eqref{CB1} over
all bipartite pure states $\psi_{A'A}$ is equivalent to 
optimizing \eqref{CB2} over all $Y\in\B(\hil_{A'})_+$ such that $\Tr \{Y\}=1$, i.e., all states $Y$ on $\hil_{A'}$.
The latter yields
\begin{align}
&\sup_{Y\in\S(\hil_{A'})}\norm{
\bz Y^{\frac{1}{2\alpha}}\otimes I_B\jz(\Theta\circ\N)(\Gamma_{A'A})
\bz Y^{\frac{1}{2\alpha}}\otimes I_B\jz}_{\alpha}\nonumber\\
&\ds=
\sup_{Y\in\B(\hil_{A'})_+\setminus\{0\}}\frac{1}{(\Tr \{Y\})^{\frac{1}{\alpha}}}\norm{
\bz Y^{\frac{1}{2\alpha}}\otimes I_B\jz(\Theta\circ\N)(\Gamma_{A'A})
\bz Y^{\frac{1}{2\alpha}}\otimes I_B\jz}_{\alpha}\\
&\ds=
\sup_{U}\sup_{Y\in\B(\hil_{A'})_+\setminus\{0\}}\frac{1}{(\Tr\{Y\})^{\frac{1}{\alpha}}}\norm{
\bz UY^{\frac{1}{2\alpha}}\otimes I_B\jz(\Theta\circ\N)(\Gamma_{A'A})
\bz Y^{\frac{1}{2\alpha}}U^*\otimes I_B\jz}_{\alpha}\label{CB3}\\
&\ds=
\sup_{Z\in\B(\hil_{A'})\setminus\{0\}}\frac{1}{(\Tr \{(Z^*Z)^{\alpha}\})^{\frac{1}{\alpha}}}\norm{
\bz Z\otimes I_B\jz(\Theta\circ\N)(\Gamma_{A'A})
\bz Z^*\otimes I_B\jz}_{\alpha}\label{CB4}\\
&\ds=
\sup_{\ket{z}\in\hil_{A'}\otimes\hil_{A}\setminus\{0\}}\frac{\norm{
(\Theta\circ\N)\pr{z}}_{\alpha}}{\norm{\Tr_{A'}\{\pr{z}\}}_{\alpha}}
=\norm{\Theta\circ\N}_{\mathrm{CB},1\to\alpha},\label{CB5}
\end{align}
where the first supremum in \eqref{CB3} is taken over all unitaries $U$ on $\hil_{A'}$,
and
the second equality in \eqref{CB5} is due to \eqref{cb1toalpha}.
This finishes the proof.
\ds\qed
\bigskip

To prove Lemma \ref{MI limits}, we will need the following minimax theorem from \cite[Corollary A2]{MH}:
\begin{lemma}\label{lemma:minimax}
Let $X$ be a compact topological space, $Y$
be a subset of the real line, and let $f:\,X\times Y\to\bR\cup\{-\infty,+\infty\}$ be a function.
Assume that
\begin{enumerate}
\item
$f(.,y)$ is lower semicontinuous for every $y\in Y$, and
\item
$ f(x,.)$ is monotonic increasing for every $x\in X$, or
$ f(x,.)$ is monotonic decreasing for every $x\in X$.
\end{enumerate}
Then 
\begin{align}
\inf_{x\in X}\sup_{y\in Y}f(x,y)=
\sup_{y\in Y}\inf_{x\in X}f(x,y).
\end{align}
\end{lemma}

It is easy to see that for any fixed $\rho,\sigma$, $\ep\mapsto D_{\alpha}(\rho\|\sigma+\ep I)$
and $\ep\mapsto \wt D_{\alpha}(\rho\|\sigma+\ep I)$ are monotone decreasing on $(0,+\infty)$, and 
\begin{align}
D_{\alpha}(\rho\|\sigma)=\sup_{\ep>0} D_{\alpha}(\rho\|\sigma+\ep I),\ds\ds\ds
\wt D_{\alpha}(\rho\|\sigma)=\sup_{\ep>0} \wt D_{\alpha}(\rho\|\sigma+\ep I);
\end{align}
for these latter relations see, e.g. \cite{MH} and \cite{MDSFT13}. Since for any fixed $\ep>0$,
$(\rho,\sigma)\mapsto D_{\alpha}(\rho\|\sigma+\ep I)$ and 
$(\rho,\sigma)\mapsto \wt D_{\alpha}(\rho\|\sigma+\ep I)$ are continuous, we get that 
\begin{align}\label{lsc}
(\rho,\sigma)\mapsto D_{\alpha}(\rho\|\sigma) \ds\ds\text{and}\ds\ds 
(\rho,\sigma)\mapsto \wt D_{\alpha}(\rho\|\sigma)\ds\ds\text{are lower semicontinuous}.
\end{align}

Now we are ready to prove Lemma \ref{MI limits}.
\medskip

\noindent\textbf{Proof of Lemma \ref{MI limits}.}\ds 
The assertions about monotonicity are obvious from the definitions and \eqref{alpha mon}.

\ref{channel div limit}
We only prove the assertion for $\wt D_{\alpha}$, as the proof for $D_{\alpha}$ goes exactly the same way. 
Let $\N_1,\N_2:\,\B(\hil_A)\to\B(\hil_B)$ be channels.
By the monotonicity \eqref{alpha mon}, we have
\begin{align}
\lim_{\alpha\nearrow 1}\wt D_{\alpha}(\N_1\|\N_2)
&=
\sup_{\alpha\in (0,1)}\wt D_{\alpha}(\N_1\|\N_2)\\
&=
\sup_{\alpha\in (0,1)}\sup_{\psi_{RA}}\wt D_{\alpha}(\N_1(\psi_{RA})\|\N_2(\psi_{RA}))\\
&=
\sup_{\psi_{RA}}\sup_{\alpha\in (0,1)}\wt D_{\alpha}(\N_1(\psi_{RA})\|\N_2(\psi_{RA}))\\
& =
\sup_{\psi_{RA}}D(\N_1(\psi_{RA})\|\N_2(\psi_{RA})) \\
&=
D(\N_1\|\N_2).
\end{align}

Note that for any $\alpha>1$, $\wt D_{\alpha}(\N_1\|\N_2)=+\infty\iff  D(\N_1\|\N_2)=+\infty$, and hence for the rest we assume that all these quantities are finite, since otherwise 
$\lim_{\alpha\searrow 1}\wt D_{\alpha}(\N_1\|\N_2)=D(\N_1\|\N_2)$ is trivial.
Let $\S$ denote the set of pure states on $\hil_{A'A}$, where $A'$ is a copy of $A$; then $\S$ is a compact set, and $\psi\mapsto \wt D_{\alpha}(\N_1(\psi)\|\N_2(\psi))$ is lower 
semicontinuous on $\S$ by \eqref{lsc}, while for a fixed $\psi\in\S$, $\alpha\mapsto \wt D_{\alpha}(\N_1(\psi)\|\N_2(\psi))$ is monotone increasing \eqref{alpha mon}.
Using now Lemma \ref{lemma:minimax}, we get
\begin{align}
\lim_{\alpha\searrow 1}\wt D_{\alpha}(\N_1\|\N_2)
&=
\inf_{\alpha>1}\wt D_{\alpha}(\N_1\|\N_2) \\
& =
\inf_{\alpha>1}\sup_{\psi\in\S}\wt D_{\alpha}(\N_1(\psi_{RA})\|\N_2(\psi_{RA}))\\
&=
\sup_{\psi\in\S}\inf_{\alpha>1}\wt D_{\alpha}(\N_1(\psi_{RA})\|\N_2(\psi_{RA})) \\
& =
\sup_{\psi\in\S}D(\N_1(\psi_{RA})\|\N_2(\psi_{RA})) \\
& =
D(\N_1\|\N_2).
\end{align} 
\medskip

\ref{mi limit}\ds We only prove the assertion for $I_{\alpha}(R;B)_{\rho}$, as the proof for $\wt I_{\alpha}(R;B)_{\rho}$ goes exactly the same way. 
First, we have
\begin{align}
\lim_{\alpha\searrow 1}I_{\alpha}(R;B)_{\rho}=
\inf_{\alpha>1}I_{\alpha}(R;B)_{\rho}&=
\inf_{\alpha>1}\inf_{\sigma_B}D_{\alpha}(\rho_{RB}\|\rho_R\otimes\sigma_B)\label{mi limit upper1}\\
&=
\inf_{\sigma_B}\inf_{\alpha>1}D_{\alpha}(\rho_{RB}\|\rho_R\otimes\sigma_B)\\
&=
\inf_{\sigma_B}D(\rho_{RB}\|\rho_R\otimes\sigma_B)=
I(R;B)_{\rho}.\label{mi limit upper3}
\end{align}
Next, note that by \eqref{lsc}, $D(\rho_{RB}\|\rho_R\otimes\sigma_B)$ is lower semicontinuous in $\sigma_B$ on the compact set $\S(\hil_B)$, and it is monotone increasing in $\alpha$. Hence, by Lemma \ref{lemma:minimax},
\begin{align}
\lim_{\alpha\nearrow 1}I_{\alpha}(R;B)_{\rho}=
\sup_{\alpha\in(0,1)}I_{\alpha}(R;B)_{\rho}&=
\sup_{\alpha\in(0,1)}\inf_{\sigma_B}D_{\alpha}(\rho_{RB}\|\rho_R\otimes\sigma_B)\\
&=
\inf_{\sigma_B}\sup_{\alpha\in(0,1)}D_{\alpha}(\rho_{RB}\|\rho_R\otimes\sigma_B)\\
&=
\inf_{\sigma_B}D(\rho_{RB}\|\rho_R\otimes\sigma_B)=
I(R;B)_{\rho}.
\end{align}

\ref{channel info limit}\ds
We only prove the assertion for $\wt I_{\alpha}(\N)$, as the proof for $I_{\alpha}(\N)$ goes exactly the same way. First,
\begin{align}
\lim_{\alpha\nearrow 1}\wt I_{\alpha}(\N)&=
\sup_{\alpha\in(0,1)}\wt I_{\alpha}(\N)=
\sup_{\alpha\in(0,1)}\sup_{\psi_{RA}}\wt I_{\alpha}(R;B)_{\N(\psi)}=
\sup_{\psi_{RA}}\sup_{\alpha\in(0,1)}\wt I_{\alpha}(R;B)_{\N(\psi)}\\
&=
\sup_{\psi_{RA}}\wt I(R;B)_{\N(\psi)}
=
I(\N).
\end{align}
Next, let~$\hat A$ be a fixed copy of~$A$.
Note that 
$\psi_{\hat AA}\mapsto 
\inf_{\sigma_B\in S(\hil)_{++}}\widetilde{D}_{\alpha}\left(  \mathcal{N}_{A\rightarrow B}(\psi_{\hat AA})\Vert\psi_{\hat A}\otimes\sigma_{B}\right)
=\wt I_{\alpha}(\hat A;B)_{\N(\psi)}$ is the infimum of continuous functions, and hence it is 
upper semi-continuous on the compact set of pure states on~$\hil_{\hat AA}$. On the other hand, it is monotone in $\alpha$ by \eqref{alpha mon}, and hence we can use 
Lemma \ref{lemma:minimax} and \eqref{mi limit upper1}--\eqref{mi limit upper3} to obtain
\begin{align}
\lim_{\alpha\searrow 1}\wt I_{\alpha}(\N)=
\inf_{\alpha>1}\wt I_{\alpha}(\N)
&=
\inf_{\alpha>1}\sup_{\psi_{\hat AA}}\wt I_{\alpha}(R;B)_{\N(\psi)}
=
\sup_{\psi_{\hat AA}}\inf_{\alpha>1}\wt I_{\alpha}(R;B)_{\N(\psi)}
=
\sup_{\psi_{\hat AA}} I(R;B)_{\N(\psi)}
=
I(\N).
\end{align} 
\noindent\qed
\bigskip

\noindent\textbf{Proof of Lemma \ref{lemma:geom int}.} 
Let $\hat A$ be a copy of $A$. 
By Lemma \ref{lemma:channel div repr}, we have 
\begin{align}
\wt I_{\alpha}(\N)&=
\sup_{\rho_{\hat A}\in\S(\hil_{\hat A})}\inf_{\sigma_B\in S(\hil_B)_{++}}
\wt D_{\alpha}\bz \rho_{\hat A}^{1/2}\N_{A\to B}(\Gamma_{\hat AA})\rho_{\hat A}^{1/2}\middle\|\rho_{\hat A}\otimes \sigma_B\jz.\label{geom1}
\end{align}
Let $\Gamma^{\N}\equiv \N_{A\to B}(\Gamma_{\hat AA})$. According to \cite[Lemma 3]{GW}, the R\'enyi divergence in \eqref{geom1}
can be written as 
\begin{equation}
\frac{1}{\alpha-1}\log s(\alpha)\wt Q_{\alpha}(\rho_{\hat A},\sigma_B),
\end{equation}
where
\begin{equation}
\wt Q_{\alpha}(\rho_{\hat A},\sigma_B) =
s(\alpha)\Tr\left\{\bz [\Gamma^{\N}]^{1/2}\bz \rho_{\hat A}^{\frac{1}{\alpha}}\otimes \sigma_B^{\frac{1-\alpha}{\alpha}}\jz[\Gamma^{\N}]^{1/2}\jz^{\alpha}\right\},
\end{equation}
and $s(\alpha):=-1$ for $\alpha\in(0,1)$, and $s(\alpha):=1$ for $\alpha>1$. For $\alpha\in[1/2,1)$, 
$x\mapsto x^{\frac{1-\alpha}{\alpha}}$ is operator concave on $\bR_{+}$, and $X\mapsto\Tr\{X^{\alpha}\}$ is monotone increasing 
and concave on positive semidefinite operators. Thus $\wt Q_{\alpha}(\rho_{\hat A},\sigma_B)$ is convex in $\sigma_B$.
Note that $\rho_{\hat A}\mapsto\rho_{\hat A}\otimes I_B$ is affine, and applying 
Theorem 1.1 in \cite{CL}, with
$p:=\frac{1}{\alpha},q=1$, $B=I_{\hat A}\otimes \sigma_B^{\frac{1-\alpha}{2\alpha}}[\Gamma^{\N}]^{1/2}$, to the quantity
(1.3) in \cite{CL}, we get that $\wt Q_{\alpha}(\rho_{\hat A},\sigma_B)$ is concave in $\rho_{\hat A}$. Similarly, for $\alpha>1$, $x\mapsto x^{\frac{1-\alpha}{\alpha}}$ is operator convex on $\bR_{++}$, and $X\mapsto\Tr\{X^{\alpha}\}$ is monotone increasing 
and convex on positive semidefinite operators. Thus $\wt Q_{\alpha}(\rho_{\hat A},\sigma_B)$ is convex in $\sigma_B$, and again by Theorem 1.1 in \cite{CL}, it is concave in $\rho_{\hat A}$. Hence, we can use the Kneser-Fan minimax theorem \cite{Fan,Kneser} to obtain
\begin{align}
\wt I_{\alpha}(\N)
&=
\sup_{\rho_{\hat A}\in\S(\hil_{\hat A})}\inf_{\sigma_B\in S(\hil_B)_{++}}
\frac{1}{\alpha-1}\log s(\alpha)\wt Q_{\alpha}(\rho_{\hat A},\sigma_B)\\
&=
\frac{1}{\alpha-1}\log s(\alpha)\sup_{\rho_{\hat A}\in\S(\hil_{\hat A})}\inf_{\sigma_B\in S(\hil_B)_{++}}
\wt Q_{\alpha}(\rho_{\hat A},\sigma_B)\\
&=
\frac{1}{\alpha-1}\log s(\alpha)\inf_{\sigma_B\in S(\hil_B)_{++}}\sup_{\rho_{\hat A}\in\S(\hil_{\hat A})}
\wt Q_{\alpha}(\rho_{\hat A},\sigma_B)\\
&=
\inf_{\sigma_B\in S(\hil_B)_{++}}\sup_{\rho_{\hat A}\in\S(\hil_{\hat A})}
\frac{1}{\alpha-1}\log s(\alpha)\wt Q_{\alpha}(\rho_{\hat A},\sigma_B)\\
&=
\inf_{\sigma_B\in S(\hil_B)}\sup_{\rho_{\hat A}\in\S(\hil_{\hat A})}
\frac{1}{\alpha-1}\log s(\alpha)\wt Q_{\alpha}(\rho_{\hat A},\sigma_B)\\
&=
\inf_{\sigma_B\in S(\hil_B)}\wt D_{\alpha}(\N\|\R_{\sigma_B})
\end{align}
for every $\alpha\in[1/2,+\infty)\setminus\{1\}$. The case $\alpha=1$ follows by
\begin{align}
I(\N)=I_1(\N)=\inf_{\alpha>1}\wt I_{\alpha}(\N)
=
\inf_{\alpha>1}\inf_{\sigma_B\in S(\hil_B)}\wt D_{\alpha}(\N\|\R_{\sigma_B})
&=
\inf_{\sigma_B\in S(\hil_B)}\inf_{\alpha>1}\wt D_{\alpha}(\N\|\R_{\sigma_B})\\
&=
\inf_{\sigma_B\in S(\hil_B)}\wt D(\N\|\R_{\sigma_B}),
\end{align}
where the second and the last identities are due to Lemma \ref{MI limits}.
\hfill\qed

\end{document}